\documentclass{article}

\usepackage{microtype}
\usepackage{graphicx}
\usepackage{booktabs} 
\usepackage{microtype}
\usepackage{graphicx}
\usepackage{booktabs} 

\usepackage{hyperref}
\usepackage{natbib}

\usepackage[ruled , linesnumbered]{algorithm2e}

\usepackage{url}

\usepackage{natbib}

\usepackage{amsfonts}
\usepackage{amsmath}
\usepackage{amsthm}
\usepackage{ dsfont }
\usepackage{subcaption}
\usepackage[dvipsnames]{xcolor}
\usepackage[export]{adjustbox}

\usepackage{authblk}

\DeclareMathOperator*{\argmin}{argmin} 

\newcommand{\sref}[1]{{Sec}~\ref{#1}}
\newcommand{\fref}[1]{{Fig}~\ref{#1}}
\newcommand{\eref}[1]{{Eq}~\ref{#1}}

\newcommand{\aref}[1]{{Algorithm}~\ref{#1}}

\newtheorem{theorem}{Theorem}[section]

\newtheorem{remark}{Remark}

\newtheorem{definition}{Definition}[section]

\begin{document}

\title{Federating Recommendations Using Differentially Private Prototypes}




\author[1]{M\'onica Ribero \thanks{ corresponding author:\texttt{mribero@utexas.edu}} }
\author[2]{Jette Henderson}
\author[2,3,4]{Sinead Williamson}
\author[1]{Haris Vikalo}

\affil[1]{Department of Electrical and Computer Engineering, University of Texas at Austin, Texas, USA}
\affil[2]{CognitiveScale, Austin, Texas, USA}
\affil[3]{Department of Statistics, University of Texas at Austin, Texas, USA}
\affil[4]{Department of Information, Risk and Operations Management, University of Texas at Austin, Texas, USA}

\maketitle 





\begin{abstract}
  Machine learning methods allow us to make recommendations to users in applications across fields including entertainment, dating, and commerce, by exploiting similarities in users' interaction patterns. However,  in domains that demand protection of personally sensitive data, such as medicine or banking, how can we learn such a model without accessing the sensitive data, and without inadvertently leaking private information? We propose a new federated approach to learning global and local private models for recommendation without collecting raw data, user statistics, or information about personal preferences. Our method produces a set of prototypes that allows us to infer global behavioral patterns, while providing differential  privacy guarantees for users in any database of the system. By requiring only two rounds of communication, we  both reduce the communication costs and avoid the excessive privacy loss associated with iterative procedures. We test our framework on synthetic data as well as real federated medical data and Movielens ratings data. We show local adaptation of the global model allows our method to outperform centralized matrix-factorization-based recommender system models, both in terms of accuracy of matrix reconstruction and in terms of relevance of the recommendations, while maintaining provable privacy guarantees. We also show that our method is more robust and is characterized by smaller variance than individual models learned by independent entities. 
\end{abstract}

\section{Introduction}\label{intro}

Machine learning models exploit similarities in users' interaction patterns to provide recommendations in applications across fields including entertainment (e.g., books, movies, and articles), dating, and commerce. Such recommendation models are typically trained using millions of data points on a single, central system, and are designed under the assumption that the central system has complete access to all the data. Further, they assume that accessing the model poses no privacy risk to individuals. In many settings, however, these assumptions do not hold. In particular, in domains such as healthcare, privacy requirements and regulations may preclude direct access to data. Moreover, models trained on such data can inadvertantly leak sensitive information about patients and clients. In addition to privacy concerns, when data is gathered in a distributed manner, centralized algorithms may lead to excessive memory usage and generally require significant communication resources. 

As a concrete example, consider the use of recommender systems in the healthcare domain. There, recommender systems have been used in a variety of tasks including decision support \citep{duan2011healthcare}, clinical risk stratification \citep{Hassan2010} and automatic detection of omissions in medication lists \citep{hasan2008towards}. Such systems are typically built using electronic health records (EHRs), which are subject to privacy constraints that limit the ability to share the data between hospitals. This restricts practical applications of recommender systems in healthcare settings as single hospitals typically do not have sufficient amounts of data to train insightful models. Even when training based on a single hospital's data is possible, the resulting models will not capture distributional differences between hospitals, thus limiting their applicability to other hospitals.

Recently, federated learning \citep{federatedLearning} was proposed as an algorithmic framework for the settings where the data is distributed across many clients and due to practical constraints cannot be centralized. In federated learning, a shared server sends a global model to each client, who then update the model using their local data. The clients send statistics of the local models (for example, gradients) to the server. The server updates the shared model based on the received client information and broadcasts the updated model to the clients. This procedure is repeated until convergence. Federated learning has proved efficient in learning deep neural networks for image classification \citep{federatedLearning} and text generation tasks \citep{yang2018applied,hard2018federated}.  

While federated methods address practical computing and communication concerns, privacy of the users in a federated system is potentially vulnerable. Although such systems do not share data directly, the model updates sent to the server may contain sufficient information to uncover model features and raw data information \citep{Milli2019modelReconstruction,koh2017understanding,Carlini:2019,Hitaj:2017:DMU}, possibly leaking information about the users. 
These concerns motivate us to adopt differential privacy \citep{dwork2006calibrating} as a framework for limiting exposure of users' data in federated systems. A differentially private mechanism is a randomized algorithm which allows us to bound the dependence of the output on a single data point. This, in turn, translates to bounds on the amount of additional information a malicious actor could infer about a single individual if that individual were included in the training set.

While the differential mechanism presents itself as a natural solution to privacy concerns of users in federated systems, combining the two paradigms faces some major challenges. The key ones emerge due to the differences in how the two frameworks function. On the one hand, federated learning algorithms are typically iterative and involve multiple querying of the individual entities to collect up-to-date information. On the other hand, in a differential privacy setting where the information obtained in each query must be privatized via injecting noise, the total amount of noise required to be added to a query scales linearly with the number of iterations (thus reducing utility of the system and the information content) \cite{kairouz2019advances, mcmahan2018general}. 

In this paper, we present a novel differentially private federated recommendation framework for the setting where each user's data is associated with one of many entities, e.g., hospitals, schools or banks. Each entity is tasked with maintaining the privacy of the data entrusted to it against possible attacks by malicious entities. An untrusted server is available to learn centralized models and communicate (in both directions) with the individual entities. Our method learns per-entity recommender models by sharing information between entities in a federated manner, without compromising users' privacy or requiring excessive communication. Specifically, our method learns differentially private prototypes for each entity, and then uses those prototypes to learn global model parameters on a central server. These parameters are returned to the entities which use them to learn local recommender models without any further communication (and, therefore, without any additional privacy risk). 

To our knowledge, the proposed framework is the first scheme that introduces differential privacy mechanisms to federated recommendations. Unlike typical federated learning algorithms, our method requires only two global communication steps. Such a succinct communication reduces the amount of noise required to ensure differential privacy while also reducing communication overhead and minimizing the risk of communication interception. Yet despite providing differential privacy guarantees to participating entities, the framework allows each entity to benefit from data held by other entities through building its own private, uniquely adapted model. Specific contributions of the paper can be summarized as follows:

    \begin{itemize}
        \item We propose a federated recommendation framework for learning latent representation of products and services while bounding the privacy risk to the participating entities. 
        This is accomplished by estimating the column space of an interaction matrix from differentially private prototypes via matrix factorization.
        \item We enable federating recommendations under communication constraints by building in the requirement that the number of communication rounds between participating entities and the shared server is only two.
        \item We demonstrate generalizable representations and strong predictive performance in benchmarking tests on synthetic and real-world data comparing the proposed framework with individual models and conventional federated schemes that lack privacy guarantees.
    \end{itemize}

  \section{Background} \label{sec:background}
  \subsection{Recommender Systems}

The goal of recommender systems is to suggest new content to users. Recommender systems can be broadly classified in two categories: content filtering and collaborative filtering. Under the content-based paradigm,  user and/or item profiles are constructed from demographic data or item information, respectively. For example, a user profile could include age while movies could be associated with genre, principal actors, etc. With this information, similarity between users or items can be computed and utilized for recommendation via, for example, clustering or nearest neighbors techniques \citep{koren2009matrix}.  Collaborative filtering \citep{goldberg1992using} relies on  past user behaviour (ratings, views, purchases) to make recommendations, avoiding the need for additional data collection \citep{Herlocker1999Collaborative,Koren2008}. In this paper we focus on collaborative filtering, although our methodology could be extended to incorporate additional content-based information \citep{rendle2010factorization}.  Below we introduce notation and summarize relevant techniques. 

Consider a set of  $n$ users and a set of $m$ items, where each user has interacted with a subset of the items. We assume that the interactions for user $i$ can be summarized via a partially observed feedback vector $\mathbf{x}_i \in \mathds{R}^m$, and that all user-item interactions can be represented by a partially observed matrix $X\in \mathds{R}^{n\times m}$. Entries $x_{ij}$ can be in the form of explicit feedback, e.g. numerical ratings from 1 to 5, or implicit, such as binary values indicating that a user viewed or clicked on some content \citep{hu2008collaborative,Zhao2018explicit,Jawaheer2010ComparisonImplicit}. The goal is to predict items that a user would like but has not previously interacted with (i.e.,\ to predict which of the missing values in $\mathbf{x}_i$ have high values).

\subsubsection{Matrix Factorization}
Matrix factorization is a popular and effective collaborative filtering approach used in many different fields to find low dimensional representation of users and items \citep{Koren2008,koren2009matrix,matrixFac2010,mcauley2013hidden}.

A matrix factorization approach assumes that users and items can be characterized in a low dimensional space $\mathds{R}^\ell$ for some $\ell \ll \min(m, n)$, i.e., that the partially observed matrix $X$ can be approximated by $X \approx UV^T$, where $U\in \mathds{R}^{n \times \ell}$ aggregates users' representations, and $V\in \mathds{R}^{m \times \ell}$ collects items' representations. In this paper, we rely on non-negative factorization, i.e., we constrain the estimates of $U$ and $V$ to be non-negative. Such a constraint often results in more interpretable latent factors and improved performance \citep{nmf2006,mcauley2013hidden}. In this setting, $U$ and $V$ can be estimated as
\begin{equation}
    \hat{U}, \hat{V} = \argmin_{U,V \geq 0} \|X-UV^T\|_2^2 + f(U, V)
    \label{eq1}
\end{equation}
where $f(U, V)$ is a regularizer. For the remainder of this paper, we assume $f(U, V) = \lambda (\|U\|^2 +\|V\|^2)$. 

Since we only have access to a subset of the entries of $X$, (\ref{eq1}) is solved by minimizing the error over the training set of ratings $T$,
\begin{equation} \label{matFact} 
    \argmin_{U,V \geq 0}\frac{1}{N}\sum_{x_{ij} \in T} \|x_{ij}-\mathbf{u}_i\mathbf{v}_j^T\|_2^2 + \lambda (\|U\|^2 +\|V\|^2),
\end{equation}
where $\mathbf{u}_i$ denotes the $i$th row of $U$ -- i.e., the latent representation for the $i$th user -- and $\mathbf{v}_j$ is the $j$th row 
of $V$ -- i.e.,\ the latent representation for the $j$th item.

\subsection{Federated learning}\label{sec:FedLearn}
Federated learning was introduced by \citet{federatedLearning} as a framework for learning models in a decentralized manner, and originally applied to learning with neural networks. The goal of federated learning is to infer a global model without collecting raw data from participating users. This is achieved by having the users (or entities representing multiple users) locally compute model updates based on their data and share these updates with a central server. The server then updates the global model and sends it back to the users. 

While they avoid directly sharing users' data, most federated learning algorithms offer no formal guarantees that a malicious agent could not infer private information from the updates. For example, in a na\"{i}ve application of the original federated learning method \citep{federatedLearning} to a matrix-factorization-based recommender system, each entity shares parameters including a low-dimensional representation of each user, leading to a high risk of potential privacy breaches.

\subsection{Differential Privacy} \label{privacy}
Differential privacy \citep{dwork2006calibrating} is a statistical notion of privacy that bounds the potential privacy loss an individual risks by 
allowing her data to be used in the algorithm. 
  
\begin{definition}
A randomized algorithm $\mathcal{M}$ satisfies $\epsilon$-differential privacy ($\epsilon$-DP) if for any datasets  $A$ and $B$ differing by only one record and any subset $S$ of outcomes $S \in range(\mathcal{M})$, $$Pr(\mathcal{M}(A) \in S) \leq e^{\epsilon} \cdot Pr(\mathcal{M}(B) \in S).$$
\end{definition}
In other words, for any possible outcome, including any given individual in a data set can only change the probability of that outcome by at most a multiplicative constant which depends on $\epsilon$. Differential privacy has been applied to recommender systems by adding noise to the average item ratings and the item-item covariance matrix \citep{mcsherry2009differentially}. However, this approach is designed for systems wherein a centralized server needs to collect all the data to derive users and items' means and covariances. Differential privacy is more difficult to impose in iterative algorithms, such as those commonly used in federated learning scenarios, since the iterative nature of these algorithms requires splitting privacy budget across iterations, thus bringing forth technical challenges \cite{abadi2016deep,wu2017bolt,mcmahan2018general}.
\subsubsection{Differentially private prototypes}

Our design of private prototypes is motivated by the efficient differentially private $k$-means estimator for high-dimensional data introduced in \citet{pmlr-v70-balcan17a}. This algorithm  first relies on the Johnson-Lindenstrauss lemma to project the data into a lower dimensional space that still preserves much of the data's underlying structure. Then, the space is recursively subdivided, with each subregion and its corresponding centroid being considered a candidate centroid with probability that depends on the number of points in the region and the desired value of privacy $\epsilon$. The final $k$-means are selected from the candidate centroids by recursively swapping out candidates using the exponential mechanism \citep{McSherry:Talwar:2007}, where the score for each potential collection is the clustering loss. The selected candidates are mapped back to the original space by taking a noisy mean of data points in the corresponding cluster, providing $\epsilon$-DP. 

The complete algorithm is parametrized by the number of clusters $k$; the privacy budget $\epsilon$; a parameter $\delta$ such that with probability at least $1-\delta$, the clustering loss $\mathcal{L}\left( \{\Tilde{z_i}\}_{i=1}^k \right)$ associated with the $k$ centers $\{\Tilde{z_i}\}_{i=1}^k$ satisfies
$$ \mathcal{L}\left( \{\Tilde{z_i}\}_{i=1}^k \right) \leq O(\log^3n)\text{OPT} + O \left(\frac{k^2\epsilon +m}{\epsilon^2} \Lambda^2 \log^5\frac{n}{\delta}\right),$$
where OPT is the optimal loss under a non-private algorithm; $n$ is the number of data points; $m$ is the dimensionality of the data; and $\Lambda$ is the radius of a ball $\mathcal{B}(0, \Lambda) \in \mathds{R}^m$ that bounds the data. We formalize this algorithm as \texttt{private\_prototypes} in the supplementary material. 

The \citet{pmlr-v70-balcan17a} method is one of a number of differentially private algorithms for finding cluster representatives or prototypes. \citet{blum2005practical} introduced SuLQ $k$-means, where the server updating clusters' centers receives only noisy averages. Unlike the approach of \citet{pmlr-v70-balcan17a}, this algorithm does not have guarantees on the convergence of the loss function. \citet{Nissim:2007:SSS} and \citet{wang2015differentially} use a similar framework but calibrate the noise by local sensitivity, which is difficult to estimate without assumptions on the dataset \citep{Zhu2017dpsurvey}.  Private coresets have been used to construct differentially private $k$-means and $k$-medians estimators \citep{privateCoresets2017}, but this approach does not scale to large data sets. 

\section{Related Work} \label{sec:relatedWork}

To the best of our knowledge, the current paper is the first to propose learning recommender systems in a federated manner while guaranteeing differential privacy. However, a number of other approaches have been proposed for incorporating notions of privacy in recommender systems \citep{friedman2015privacy}; we summarize the main ones below.

        \paragraph{Cryptographic methods.} 
        
        A number of private recommender systems have been developed using a cryptographic approach \citep{miller2004pocketlens,kobsa2003privacy}. Such methods use encryption to protect users by encoding personal information with cryptographic functions before it is communicated.  In the healthcare context, \citet{hoens2013reliable} have applied cryptographic methods to providing physician recommendations. However, these methods require centralizing the dataset to perform calculations on the protected data, which may be infeasible when the total data size is large or communication bandwidth is limited, or where regulations prohibit sharing of individuals' data even under encryption.

\paragraph{Federated recommender systems.}

In addition to the generic federated learning algorithms discussed in \sref{sec:FedLearn}, alternative federation methods have been proposed for matrix factorization, where the information being shared is less easily mapped back to individual users. \citet{kim2017federated} consider federated tensor factorization for computational phenotyping. There, the objective function is broken into subproblems using the Alternating Direction Method of Multipliers  \citep[ADMM,][]{boyd2011distributed}, where the alternated optimization is utilized to distribute the optimization between different entities. User factors are learned locally, and the server updates the global factor matrix and sends it back to each entity. In a similar way, \citet{ammad2019federated} perform federated matrix factorization by taking advantage of alternating least squares. They decouple the optimization process, globally updating items' factors and locally updating users' factors. These two approaches converge to the same solution as non-federated methods. However, since current variables need to be shared at each optimization stage, this technique requires large communication rates and users' synchronization. While either of the above factorization methods could be adapted to recommender systems, they lack strict privacy guarantees and require extensive communication.
      
\paragraph{Differential privacy.}

In a recommender systems context, \citet{mcsherry2009differentially} rely on differential privacy results to obtain noisy statistics from ratings.  Although the resulting model provides privacy guarantees, it requires access to the centralized raw data in order to estimate the appropriate statistics. This makes it unsuited for the data-distributed setting we consider.

\paragraph{Exploiting public data.}

\citet{controllingPrivacy2014} consider matrix factorization methods to learn $X\approx UV^T$ (see \sref{sec:background}) in the setting where we can learn the item representation matrix $V$ from publicly available data. The public item matrix is then shared with private entities to locally estimate their latent factors matrix $U$. The applicability of this approach is hindered by potentially limited access to public data, which is the case in sensitive applications such as healthcare recommendations. Our approach provides an alternative method for learning a shared estimate of $V$ from appropriately obscured private data.

  \section{A Differentially Private, Federated Recommender System}\label{sec:model}
  
  We propose a model for learning a recommender system in a federated setting where different entities possess a different number of records. We assume the data is split between $H$ entities such that each entity possesses data for more than one user. The partially observed user-item interaction matrix associated with the $h$th entity is denoted by $X_h$.  
  
  We assume that the training data is sensitive and should not be shared outside the entity to which it belongs. While each entity will need to communicate information to a non-private server, we wish to ensure this communication does not leak sensitive information.
  
  In order for differential privacy and federated recommender systems to work in concert, our framework must accomplish two objectives: 1) make recommendations privately by injecting noise in a principled way, and 2) reduce the number of communications to minimize the amount of injected noise. The solutions to these requirements are interrelated. 
  We first describe a method that reduces the number of communication steps to two, and then procede to describe how to solve the privacy challenge.

  \subsection{A One-shot Private System}
  
Most federated learning methods require multiple rounds of communication between entities and a central server, which poses a problem for differential privacy requirements.
Specifically, we can think of each round of communication from the entities to the server as a query sent to the individual entities, which has potential to leak information. 
If we query an $\epsilon$-DP mechanism $K$ times, then the sequence of queries is only $K\epsilon$-DP \citep{McSherry:2009}. 
In practice, this means that, the more communication we require, the more noise must be added to maintain the same level of differential privacy.

To minimize the amount of noise a differential privacy technique will introduce, our method must limit the number of communication calls between the entities. In the context of matrix factorization-based recommendations which involve estimating $\hat{X} = \hat{U}\hat{V}^T$, as discussed in \sref{sec:relatedWork}, \citet{controllingPrivacy2014} show that transmission of private data can be avoided by using a public dataset to learn the shared item representation $\hat{V}$.
Given $\hat{V}$, each entity can privately estimate $\hat{X}_h = \hat{U}_h\hat{V}^T$ without releasing any information about $X_h$. 
Building upon this idea, we constrain the communication to only two rounds, back and forth. However, in our problem setting we do not have access to a public data set. Instead, we construct a shared item representation $\hat{V}$ based on privatized prototypes $P_h$ collected from each entity.  These prototypes are designed to: a) contain similar information as $X_h$, thus allowing construction of an accurate item representation; b) be of low dimension relative to $X_h$, hence minimizing communication load; and c) maintain differential privacy with respect to the individual users. 
We elaborate on building prototypes in \sref{sec:prototypes}.
  
Once we have generated prototypes for each entity, we send them to a centralized server that learns shared item representation $\hat{V}$ through traditional matrix factorization (see \sref{sec:background}). This shared matrix is then communicated back to the individual entities which use it to learn their own users' profile matrices and make local predictions. 

In contrast to iterative methods, the proposed approach requires only two rounds of communication: one from the entities to the server, and one from the server to the entities. In addition to reducing communication costs and removing the need for synchronization, this strategy allows us to conserve the privacy budget. With only one communication requiring privatization, we are able to minimize the noise that must be added to guarantee a desired level of privacy.

    \subsection{Learning Prototypes}\label{sec:prototypes}

For our algorithm to succeed, we must find a way to share the information from all $H$ entities in order to build a global item representation matrix $\hat{V}$. We want the prototypes to be representative of the data set, i.e., ensure they convey useful information. Note that to satisfy $\epsilon$-differential privacy, each set of prototypes must be $\epsilon$-differentially private.

Differentially private dataset synthesis methods (see \citet{Bowen:Liu:2016} for a survey) could be used to generate $\widetilde{X}_h$ having statistical properties similar to $X_h$. However, these methods tend to be ill-suited for high-dimensional settings and would involve sending a large amount of data to the server. Instead, we consider methods that find differentially private \textit{prototypes} of our dataset, with the aim of obtaining fewer samples that still capture much of the variation present in the individual data. Since we will use these prototypes to capture low-rank structure, provided each entity sends the number of prototypes larger than the rank, it is possible for such prototypes to contain the information required to recover singular vectors of $X_h$ yet still be smaller than $X_h$, thus reducing the amount of information that needs to be communicated. When selecting the prototype mechanism, we recall the following two observations.

\begin{algorithm*}[h!]
\SetAlgoLined
\KwIn{Data belonging to same cluster $\{x_i\}_{i=1}^m \subseteq \mathbb{R}^{n}$ with $\|x_i\|_{\infty} \leq \Lambda$, $\|x_i\|_{0} \leq s$; parameters $ \epsilon, \delta$}
\KwResult{ Centroid $v \in \mathbb{R}^{n}$}

Compute $\mu = \frac{1}{n} \sum_{i=1}^nx_i$ \\
$\Tilde{ \mu} \gets \mu + Y$ with $Y \sim \text{Gumbel}(\frac{\epsilon}{2s\Lambda}|\mu|)$ \\
 $I\gets \{i : \Tilde{\mu}_i \text{ is in top } s \}$ \\
$v \gets \begin{cases}
          \mu_i  + Lap\left(\frac{2\Lambda s}{\epsilon n} \right) & \text{ if } i \in I \\
        0 & \text{ otherwise}
        \end{cases}$ \\
\Return{v}
 \caption{Sparse recovery in high dimension}
 \label{alg:sparserec}
\end{algorithm*}

\begin{algorithm*}[h!]
    \SetAlgoLined
    \KwIn{Per-entity ratings matrices $\{X_h \in \mathcal{B}(0, \Lambda)\}_{h=1}^H$; parameters $k, \epsilon, \delta, \lambda$}
    \KwResult{ Shared $n\times \ell$ item matrix $\hat{V}$, private $n_h\times \ell$ user matrices $\hat{U}_h$, private reconstructions $\hat{X}_h$ }

    \For{$h \in \{1, \dots, H\}$} {
        $P_h \gets$ \texttt{private\_prototypes}$(X_i, \epsilon, \delta, k)$ \\
        Send $P_h$ to server \\
    }
    Compile prototypes, $P = [P_1^T, P_2^T, \cdots, P_H^T]^T$ \\
    Estimate $\hat{V}$ from $P$ following \eref{matFact}\\
    Broadcast $\hat{V}$ to all entities.\\
    \For{$h \in \{1, \dots, H\}$} {
       Estimate $\hat{U}_h$ given $\hat{V}$ following \eref{matFact}. \\
    Predict $\hat{X}_h = \hat{U}_h\hat{V}^T$ \\
        }
 \caption{Federated Recommender System}
 \label{alg:fedRecSys}
\end{algorithm*}


\begin{remark}
Non-negative matrix factorization (NMF) and spectral clustering have been shown to be equivalent \citep{ding2005equivalence}. 
\end{remark}

\begin{remark}
(Theorem 3 in \citet{pollard1982quantization})
Let $m_1, \dots, m_k$ be the optimium of the $k$-means objective on a dataset $X = \{x_i\}_{i=1}^n$ distributed according to some distribution $P$ on $\Omega$, and let $\mathcal{M}_k$ be the set of discrete distributions on $\Omega$ with support size at most $k$. Then, the discrete distribution implied by $m_1, \dots, m_k$ is the closest  discrete distribution to $P$ in $\mathcal{M}_k$ with respect to the 2-Wasserstein metric.
\end{remark}

Since we are learning the item matrix $\hat{V}$ via NMF, Remark 1 suggests that one should capture the centroids of clusters in $X_h$ to preserve spectral information. Remark 2 implies that the prototypes obtained via $k$-means are close, in a distributional sense, to the underlying distribution. 
Following these intuitions, we consider prototype generation methods based on $k$-means. 
Since the learned prototypes are created to capture the same latent representation that would be captured by NMF, we expect the estimated item matrix $\hat{V}$ to be close to the true $V$. 

Due to being appropriate for high-dimensional data, we adopt the framework of the differentially private candidates algorithm of \citet{pmlr-v70-balcan17a}. 
Note that this algorithm initially maps the data onto a low-dimensional space; however, since we are using the prototypes to learn a low-dimensional representation, such a mapping is unlikely to adversely impact the accuracy of the proposed method. We  augment the scheme of \citet{pmlr-v70-balcan17a} -- while maintaining accuracy and privacy guarantees -- by a novel recovery algorithm. The algorithm increases overall efficiency by exploiting the sparsity of the data and the Gumbel trick, often used to efficiently sample from discrete distributions \cite{papandreou2011perturb, balog2017lost,durfee2019practical}. 

 After obtaining cluster assignments for each datapoint, instead of sequentially applying the exponential mechanism to recover non-zero entries on the centroid, we add noise drawn from a Gumbel distribution to the centroid mean and take the top-$s$ entries, where $s$ denotes the number of non-zero entries in the dataset. We formalize this procedure as \aref{alg:sparserec}. \aref{alg:fedRecSys} summarizes the proposed private, federated recommender system.

\begin{theorem}
\aref{alg:fedRecSys} is $\epsilon$-Differentially Private.
\end{theorem}
\begin{proof} The server interacts with the private datasets $X_h$ only once, when collecting the private prototypes. \citet{durfee2019practical} prove that adding noise  $Y \sim \text{Gumbel}(\frac{2\Delta q}{\epsilon})$ to the utility function $q$, and selecting the top $k$ values from the noisy utility, is equivalent to applying the exponential mechanism $k$ times; therefore, transmission of a single prototype is  $\epsilon$-DP.  The parallel composition theorem \citep{McSherry:2009} establishes that the overall privacy budget is given by the maximum of the individual budgets, implying that the overall algorithm is $\epsilon$-DP. 
\end{proof}

 \section{Experiments} \label{sec:experiments}
  
  We first test the performance of the proposed differentially-private federated recommender system on synthetic data and report the results in 
  \sref{sec:exp_prototypes}. Then, to demonstrate the ability to provide high-quality recommendations in realistic settings, in \sref{sec:real_data} we apply the system to real-world datasets.

    For all the experiments, we fixed the level of regularization to $\lambda= 0.1 $ since we did not observe notable difference in performance 
    when varying $\lambda$ from $0.01$ to $10$.

  \subsection{Datasets}\label{sec:datasets}
  We test the proposed scheme on three different datasets. The first one is a synthetic dataset intended to simulate discrete processes such as ratings or counting event occurrences. The relevant matrices are generated as $U \sim \text{Norm}(0,1) \in \mathds{R}^{m\times \ell}$, $V \sim \text{Norm}(0,1) \in \mathds{R}^{n\times \ell}$,  and $X \sim \text{Pois} (\exp(UV^T)) $. We set $n = 100,000$, $m=500$, $\ell = 100$, and  distribute the data uniformly across 10 different entities.

   The second dataset is from the eICU Collaborative Research Database \citep{pollard2018eicu}, which contains data collected from critical care units throughout the continental United States in 2014 and 2015. Since different visits can have diverse causes and diagnoses, we count each patient visit as a separate observation. We  use the $\texttt{laboratories}$ and $\texttt{medicines}$ tables from the database, and create a 2-way table where each row represents a patient and each column either a lab or medicine. Matrix $X$ is composed using data from over $190$k patients, $457$ laboratories and medications, and $205$ hospitals. Each entry $x_{ij}$ represents how many times a patient took a test or a medication; the goal is to recommend treatments.
   
   Finally, we consider the Movielens 1M dataset, containing 1,000,209 anonymous ratings from 6,040 MovieLens users on approximately 3,900 movies. We use the first digit of each user's ZIP code to set up a natural federation of the data.

\subsection{Evaluation metrics}
To assess convergence and perform parameter tuning, we use the Root Mean Squared Error (RMSE) between the real $X$ and the reconstructed $\hat{X} = \hat{U}\hat{V}^T$. In the case of the synthetic data, RMSE is a suitable measure to examine the fit quality since we have access to the ground truth.

Additionally, to evaluate the quality of recommendations in the hospital and movie data tests, we compare the real and predicted rankings over the test samples using Mean Average Ranking~\cite{hu2008collaborative}. Concretely, let $\texttt{rank}_{ui}$ be the percentile of the predicted ranking of item $i$ for user $u$, where 0 means very highly ranked and above all other items. We calculate  $\overline{\texttt{rank}}$ on a test set $\mathcal{T}$ defined as 
\begin{equation}\overline{\texttt{rank}} =\frac{ \sum_{(u,i): x_{ui} \in \mathcal{T}} x_{ui} \texttt{rank}_{ui} }{\sum_{(u,i): x_{ui}\in \mathcal{T}} x_{ui} }.\label{eqn:rank}\end{equation}
This measure compares the similarity between the real and predicted ranks. Intuitively, for a random ranking the expected $\texttt{rank}_{ui}$ is $0.5$, so $\overline{\texttt{rank}} \geq 0.5$ means a ranking no better than random. Conversely, lower values indicate highly ranked recommendations matching the users' patterns.

   \subsection{Evaluating the impact of federation and privacy on synthetic data}\label{sec:exp_prototypes}
   
   Recall that our algorithm differs from the standard matrix factorization schemes in two key aspects: first, it learns the item matrix $\hat{V}$ using \textit{prototypes}, rather than the actual data; second, it learns the users' sub-matrices $\hat{U}_h$ independently given $\hat{V}$, rather than jointly. Moreover,insteat of learning the prototypes using exact $k$-means, to ensure differential privacy we use an $\epsilon$-DP algorithm. Here we explore the effect of these algorithmic features.

  In particular, we compare our framework with the following algorithms:
   \begin{itemize}
       \item \textbf{Matrix factorization}: Apply Eq~\ref{matFact} until convergence on $X = [X_1^T, \dots, X_H^T]^T$.
       \item \textbf{MF + $k$-means}: Apply \eref{matFact} to factorize a matrix of exemplars $P\approx \hat{U}\hat{V}^T$, where $P$ collects the $k$-means from each matrix $X_1, \dots, X_H$. Use the estimate $\hat{V}$ to learn individual matrices $\hat{U_h}$ from $X_h$. 
       \item \textbf{MF + $k$-random}: Identical to \textbf{MF + $k$-means}, but instead of using the cluster means, use $k$ random samples from $X_1, \dots, X_H$.
       \item \textbf{MF + $\epsilon$-private prototypes}: Identical to \textbf{MF + $k$-means}, but instead of using true cluster means, use 
       the generated $\epsilon$-DP prototypes \footnote{See Algorithm \texttt{private\_prototypes} in the supplementary.} .
   \end{itemize}

   \begin{figure*}[!h]
\centering
     \begin{subfigure}{0.22\textwidth}
     \centering
\includegraphics[width=\textwidth,left]{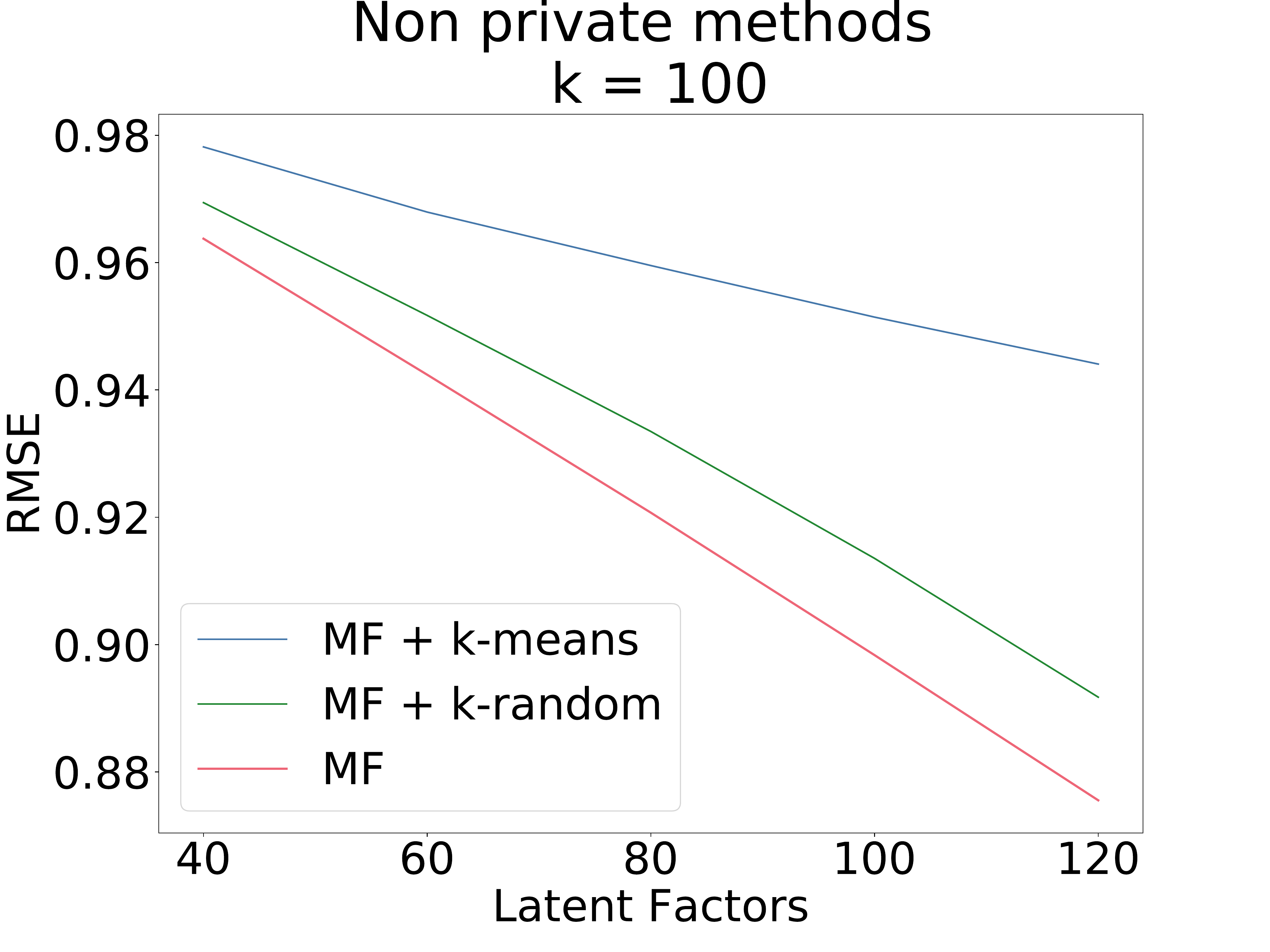}
\caption{RMSE vs. latent factors for non-private methods.}
\label{fig:exp2:nonprivate1}
\end{subfigure}
     \begin{subfigure}{0.22\textwidth}
     \centering
\includegraphics[width=\textwidth,left]{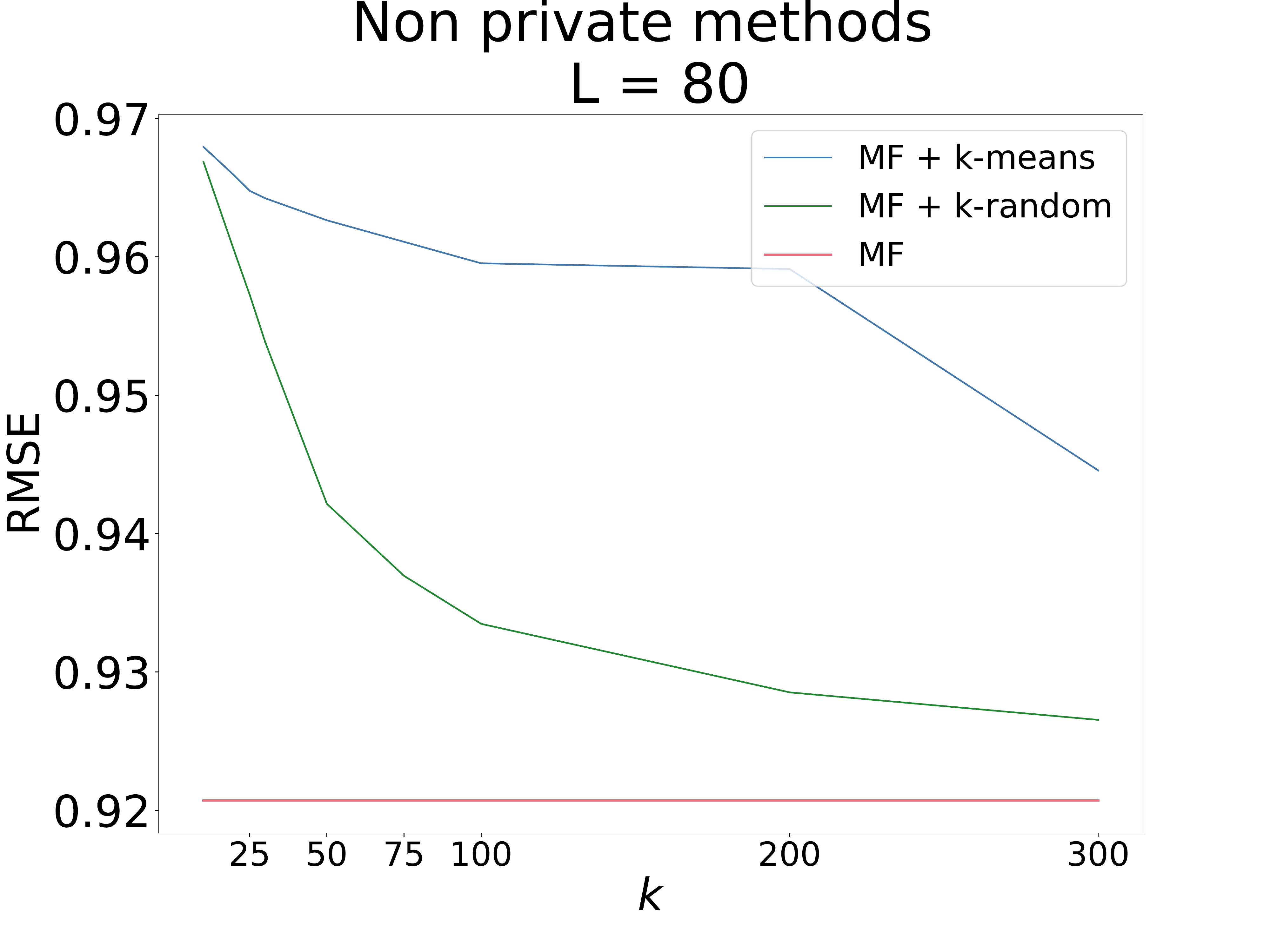}
\caption{RMSE vs. $k$ for non-private methods.}
\label{fig:exp2:nonprivate2}
\end{subfigure}
     \begin{subfigure}{0.22\textwidth}
     \centering
\includegraphics[width=\textwidth,left]{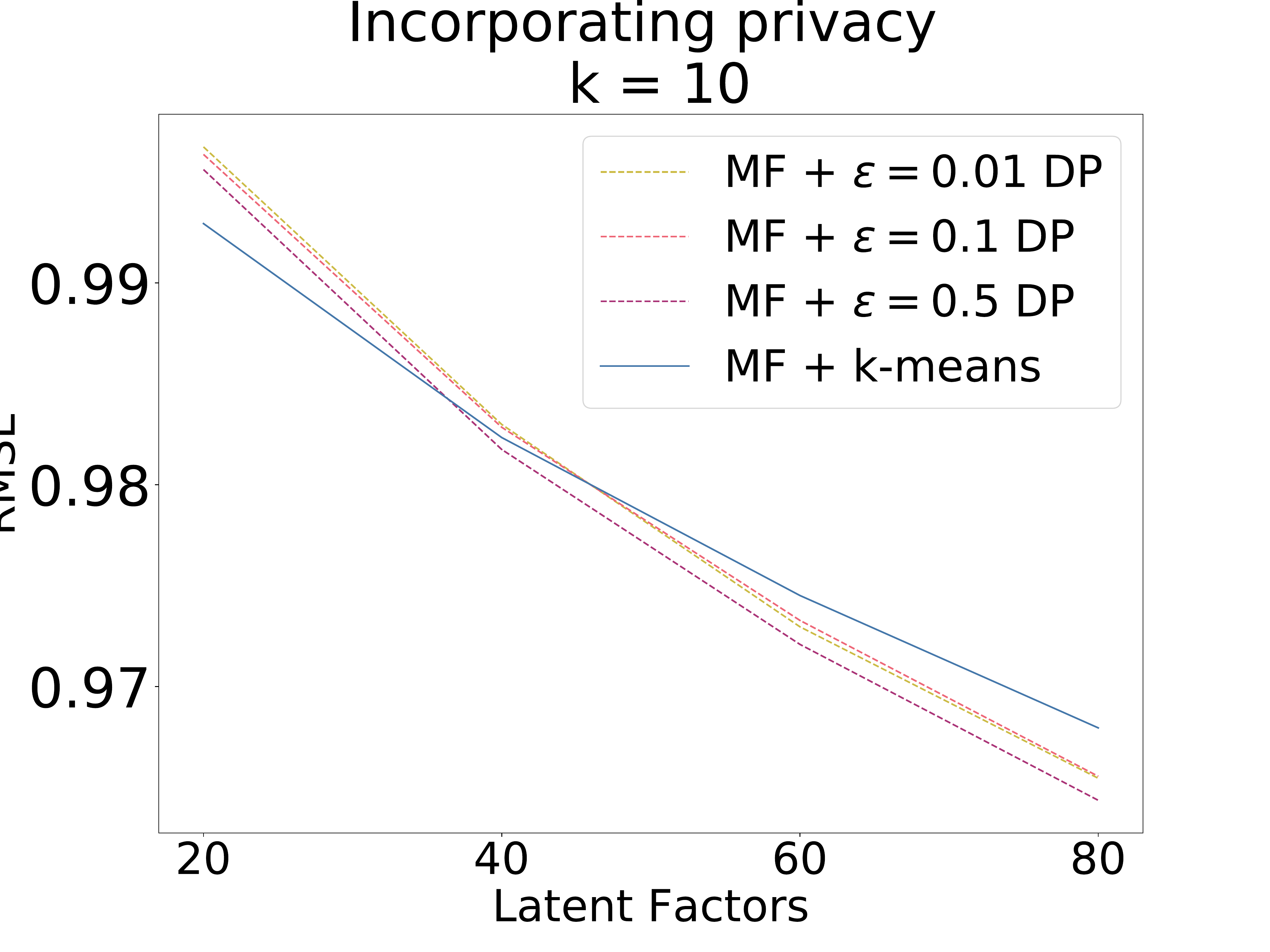}
\caption{RMSE vs. latent factors for private methods.}
\label{fig:exp2:private3}
\end{subfigure}
     \begin{subfigure}{0.22\textwidth}
     \centering
\includegraphics[width=\textwidth,left]{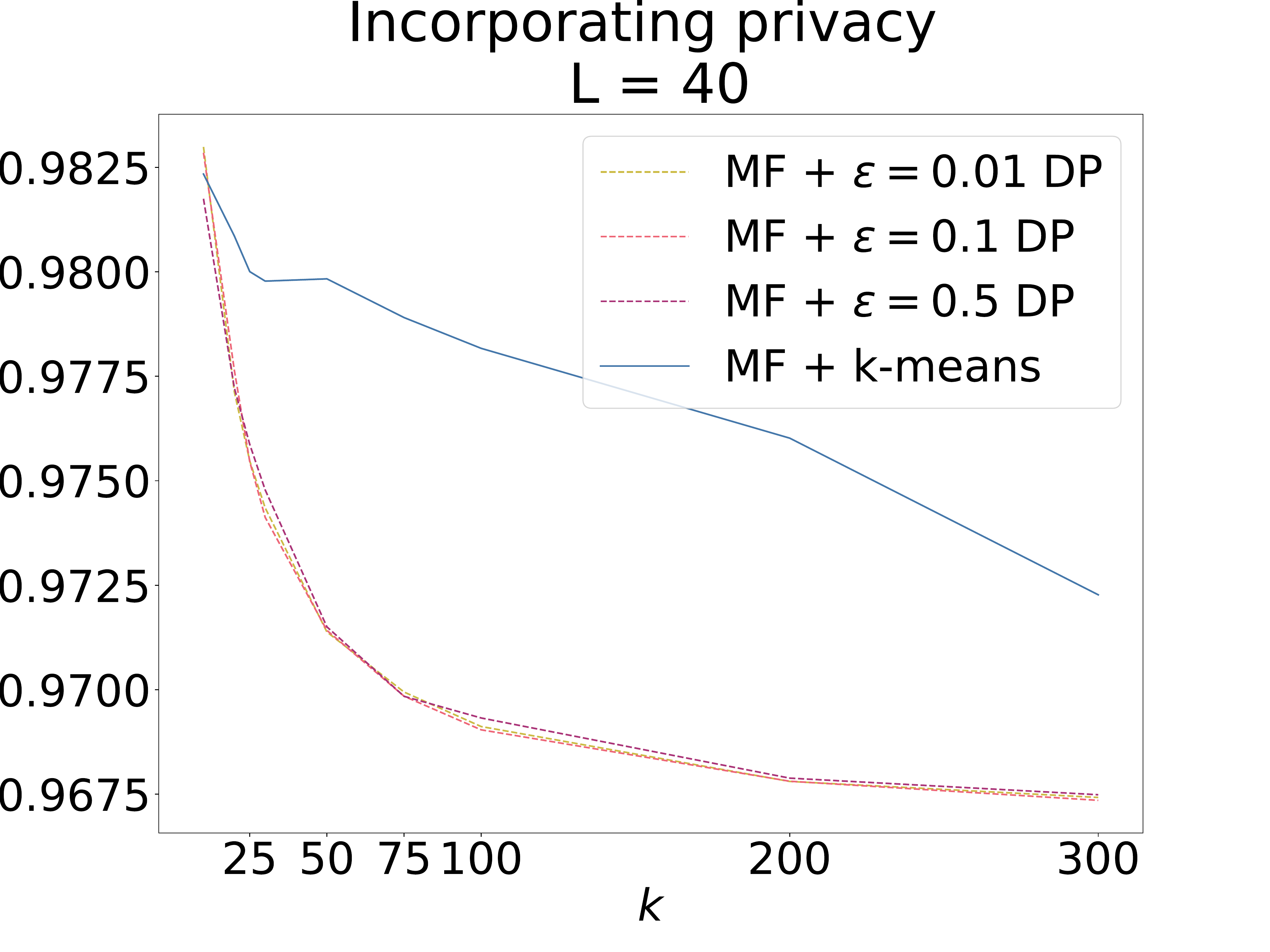}
\caption{RMSE vs. $k$ for private methods.}
\label{fig:exp2:private4}
\end{subfigure}
\caption{Results on synthetic data}
\end{figure*}

   \begin{figure*}[!h]
\centering
     \begin{subfigure}{0.31\textwidth}
         \centering
         \includegraphics[width=\textwidth]{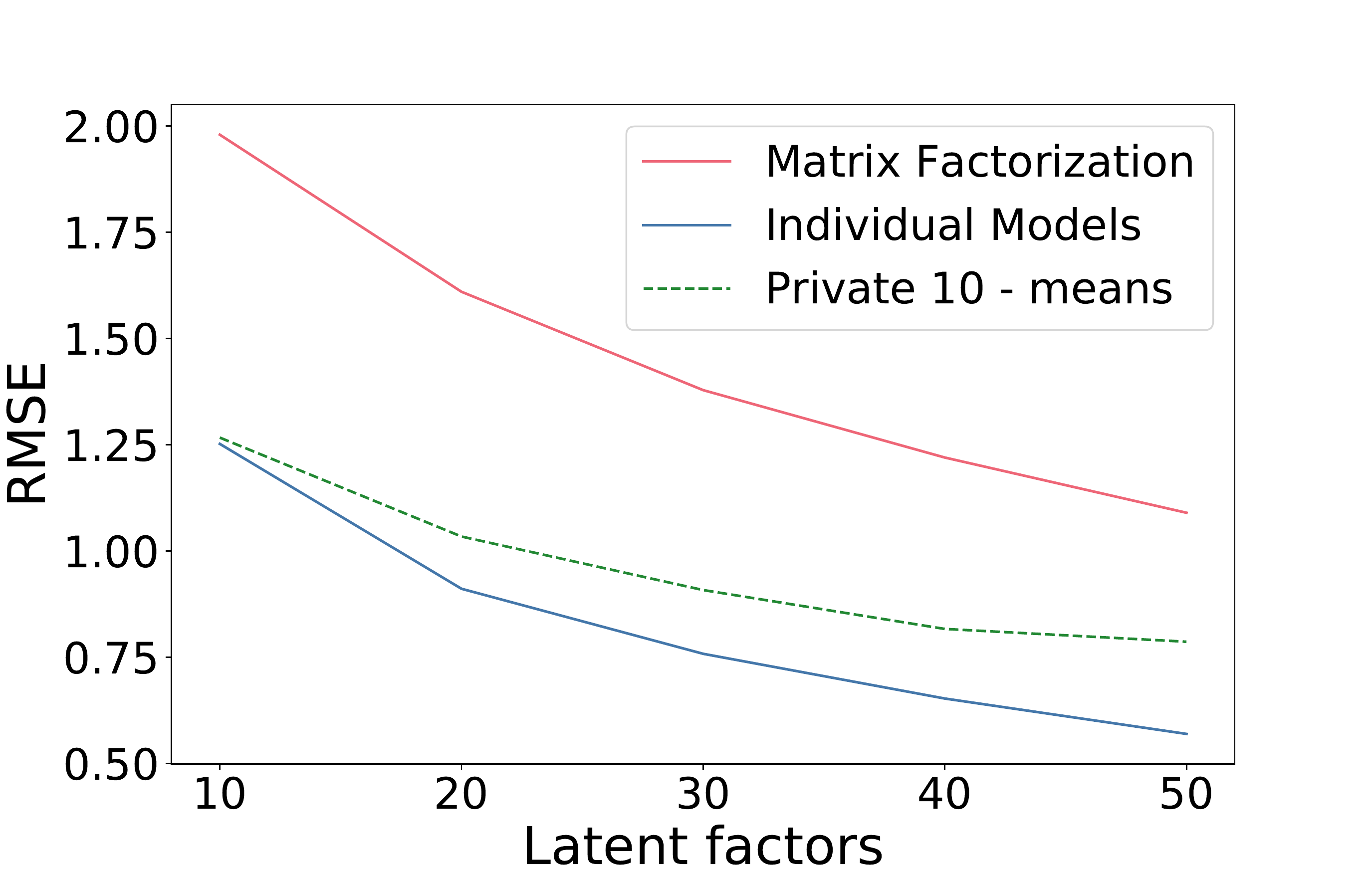}
        \caption{Training set RMSE}
        \label{fig:fed_learn}
     \end{subfigure}
     \hfill
     \begin{subfigure}{0.31\textwidth}
         \centering
         \includegraphics[width=\textwidth]{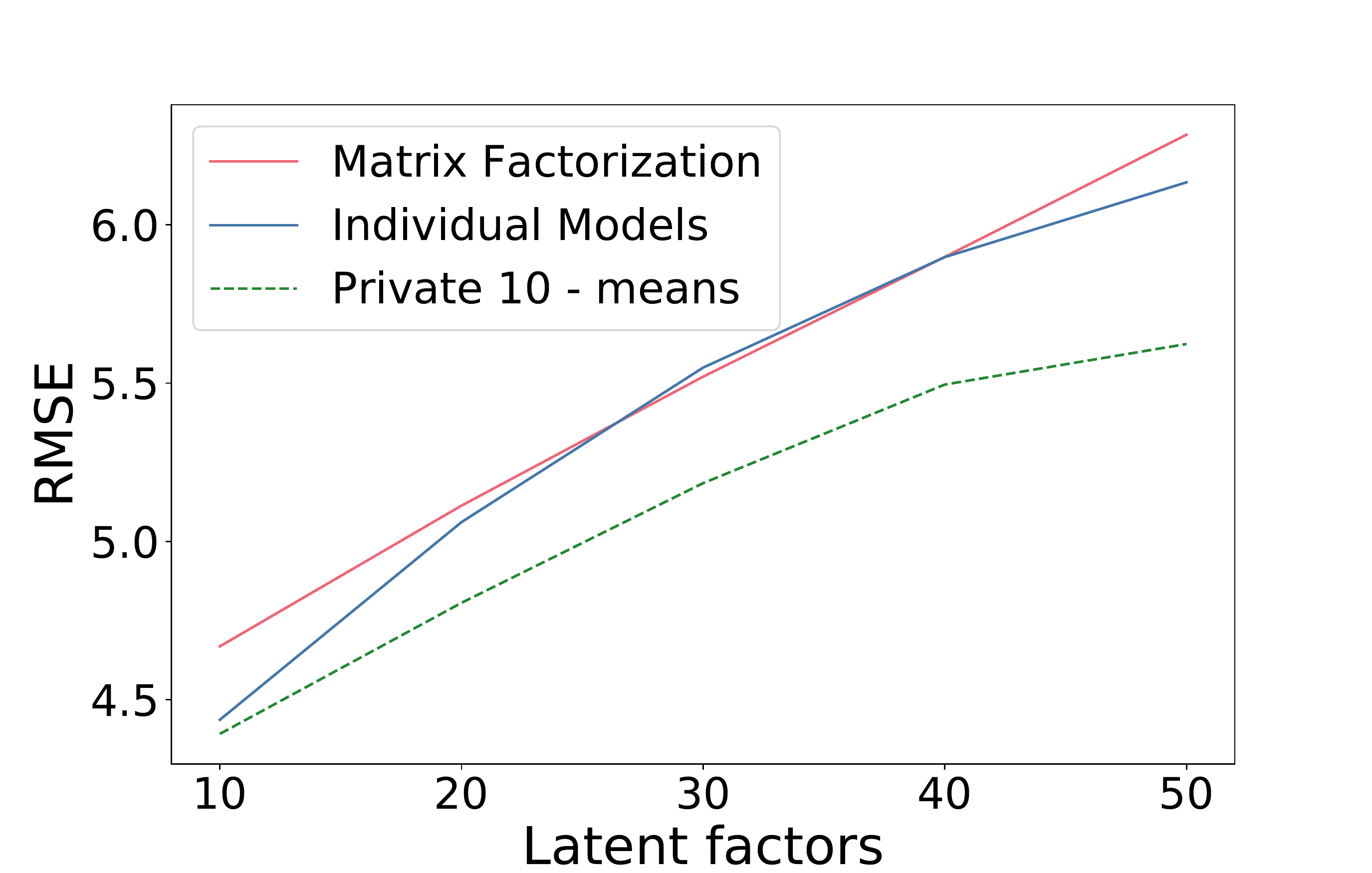}
\caption{Test set RMSE}
\label{fig:fed_learn_test_1}
     \end{subfigure}
     \hfill
     \begin{subfigure}{0.31\textwidth}
         \centering
         \includegraphics[width=\textwidth]{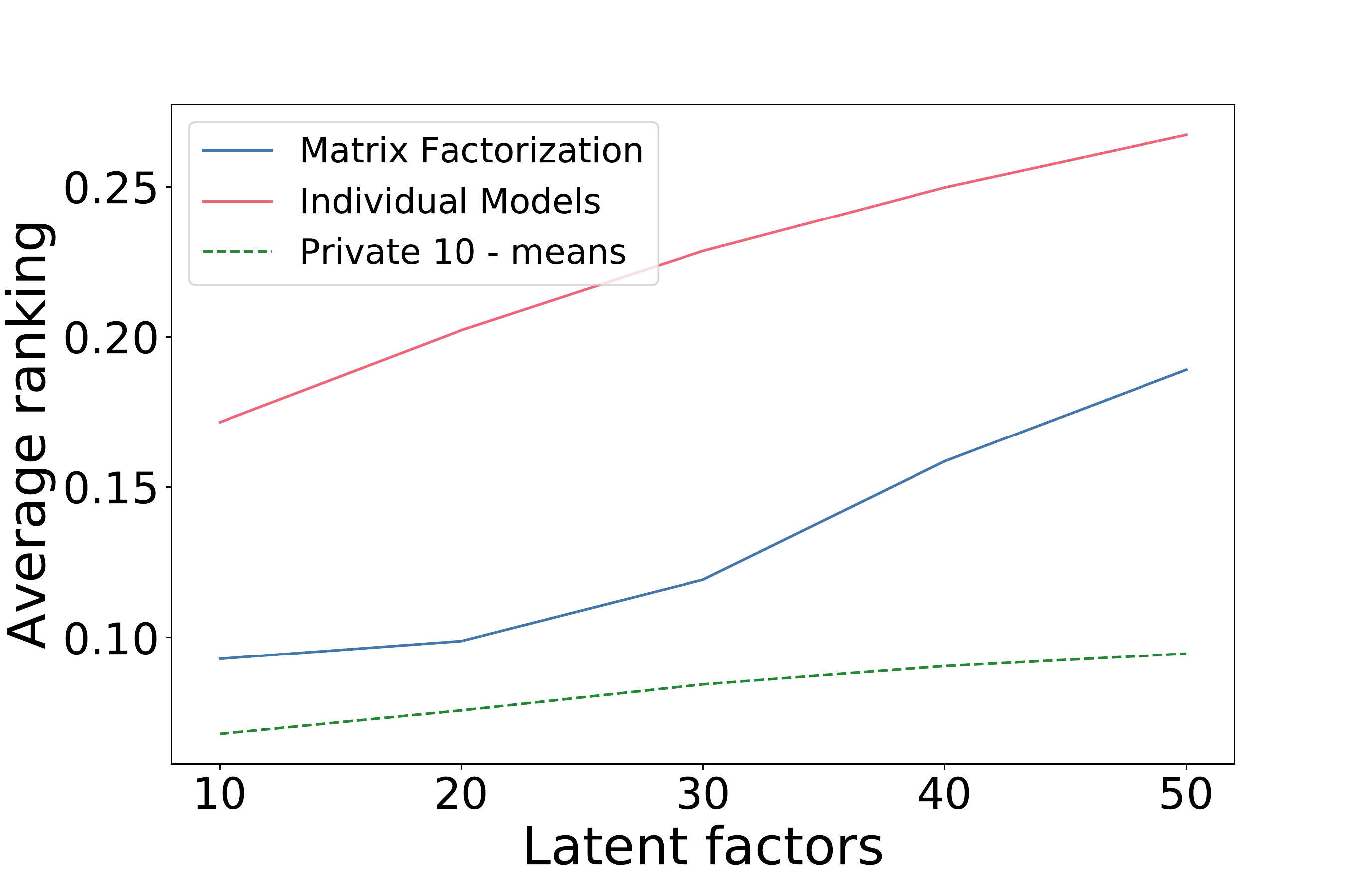}
\caption{$\overline{\texttt{rank}}$ (\eref{eqn:rank})}
\label{fig:avg_rank}
     \end{subfigure}
     \caption{Comparison of three different methods on the eICU dataset.}
     \label{fig:eICU}
\end{figure*}

We first evaluate how $k$-means performs in a non-private setting. Figures \ref{fig:exp2:nonprivate1} and \ref{fig:exp2:nonprivate2} show the RMSE when $k$ and $\ell$ are fixed, respectively.\footnote{In the supplementary, we provide additional experiments varying  $k$ between 10 and 300, and $\ell$, the dimension of the latent space, between 20 and 80l the behavior is similar.}
In both figures, we see unsurprisingly that MF has the lowest RMSE, with $k$-random exemplars from the original dataset performing second best.
For larger values of $k$ in \fref{fig:exp2:nonprivate2}, $k$-means performance deteriorates compared to $k$-random.
Based on examination of the centroids, this is most likely due to $k$-means overfitting to outliers for larger values of $k$ while $k$-random performance improves as its number of exemplars approaches the full $X$. We note that our synthetic data does not contain any clusters, so this is the worst-case scenario for the $k$-means setting; even so, we observe that the difference in reconstructive performance between the three methods is fairly small. None of the above methods guarantee privacy.
  
   Next, we compare the performance of private $k$-means and non-private $k$-means.
   In \fref{fig:exp2:private3}, we consider a relatively small value of $k=10$, and investigate the effect of $\epsilon$ as the number of latent factors changes. As expected, larger values of $\epsilon$ (i.e.,\ less private settings) yield better results. Here we observe little difference in the performance between the private and non-private algorithms. However, in \fref{fig:exp2:private4} we see that for large $k$, the private methods perform significantly better than the non-private $k$-means, mirroring the results in \fref{fig:exp2:nonprivate2}. We hypothesise that the noise introduced in the private and random scenarios acts as a regularizer, helping avoid overfitting. We note that, since the sensitivity of the random exemplar mechanism is equal to the range of the data, directly privatizing random exemplars would add excessive noise.

In both \fref{fig:exp2:private3} and \fref{fig:exp2:private4}, we find that decreasing $\epsilon$ (and therefore increasing privacy) does not have a significant negative effect on the reconstruction quality. In \fref{fig:exp2:private4}, for larger values $k$, MF + private $k$-means performs equally well, even for the smallest value of $\epsilon$, as the noise is averaged over a large number of samples. Here,  we can guarantee $0.01$-DP instead of $0.5$-DP with a minimal drop in RMSE.

   \subsection{Evaluating the federated recommender system}\label{sec:real_data}
To evaluate the entire system, we assess our model on real-world data from the eICU dataset and the Movielens 1M dataset. Similar to the experiments in the previous section, we assume that each entity extracts exemplars via private prototypes and sends them to the server. The server learns the item matrix $\hat{V}$ and sends it back to the entities. Each entity learns its own user matrix $\hat{U}_h$ and reconstructs $\hat{X}_h = \hat{U}_h\hat{V}^T$. We construct a test set $\mathcal{T}$ by randomly selecting 20\% of the users; for each selected user, we randomly select five entries.

We compare our private federated recommender system with: 1) non-private centralized matrix factorization, and 2) individual centralized matrix factorization for each hospital. The comparison is performed using different numbers of latent factors $\ell$ varied between 10 and 50. For the private prototypes, we fix $k=10$ for the hospitals' data and $k=50$ for Movielens; in both cases, $\epsilon = 0.1$.

The results on two datasets are similar and thus we present Movielens results in the supplementary. \fref{fig:fed_learn} and \fref{fig:fed_learn_test_1} show the average reconstruction error over training and test data, respectively. As expected, on the training set we see lower RMSE by the individual models than when using a jointly learned model, since there are $H$ times as many parameters to model the overall variation. Perhaps surprisingly given the noise introduced via the differential privacy mechanism, the federated model achieves training set RMSE comparable to that achieved by individual models.

Analysis of the test set RMSE (\fref{fig:fed_learn_test_1}) reveals the benefit of the federated model. The individual models obtain RMSE comparable to the jointly learned model, indicating that the low training set RMSE results from the individual models overfitting. The federated model, however, generalizes well to the test set. We hypothesise that this is because the jointly learned item matrix aids in generalization, and the use of noisy prototypes discourages overfitting.

\fref{fig:avg_rank} shows the average ranking quality for the three methods. Consistent with the test set RMSE, the federated model obtains the best ranking performance. As intended, the federated model allows each hospital to improve its predictions by obtaining relevant information from other hospitals, without compromising its patients' information.

\section{Conclusion} \label{sec:conclusion}
We propose a novel, efficient framework to learn recommender systems in federated settings. Our framework enables entities to collaborate and learn common patterns without compromising users' privacy, while requiring minimal communication. Our method assumes individuals are grouped into entities, at least some of which are large enough to learn informative prototypes; we do not require privacy within an entity. 

A future direction could be to extend this approach to the more extreme scenarios where each entity represents a single individual. This would be useful for commerce or content sites where each user wants to maintain privacy. 
Another avenue for future work is to investigate error bounds for the reconstructed matrix. Such results could allow entities to determine an appropriate privacy budget while still learning useful models.

\newpage
\bibliography{references}

\newpage
\appendix
\section{Private $k$-means Definitions and Subroutines}

\subsection{Differential Privacy Definitions}

\begin{definition}
Let $q: \mathcal{X} \times \mathcal{Y} \rightarrow \mathbb{R}$ be a utility function where $q(X,y)$ measures the utility of outputting $y$ given a dataset $X$. The \textbf{exponential mechanism} outputs $y$ with probability proportional to $\exp(\frac{\epsilon q(X,y)}{2 \Delta q})$, where $\Delta$ is the sensitivity of $q$ defined by $\Delta q = \sup_{D,D', y} |q(D,r) - q(D',r)|.$
\end{definition}

\begin{definition}
A random variable $Y$ follows a Gumbel distribution with parameter $b$ if its PDF is given by $p(y;b) = \frac{1}{b}\exp \left(-(y/b + e^{-y/b})\right)$. 
\end{definition}

\subsection{Subroutines}
\begin{algorithm*}
\SetAlgoLined
\KwIn{data $X \in \mathcal{B}(0,\Lambda) \subseteq \mathbb{R}^{ n \times p}$, parameters $ k, \epsilon, \delta$}
\KwResult{cluster centers $z_1, z_2, ..., z_k \in \mathds{R}^m$}
    Set latent dimension $p=8 \log n$, number of trials $T = 2 \log\frac{1}{\delta}$\\
    \For{$t = 1,...,T$} {
    Randomly project data from $\mathds{R}^m\rightarrow\mathds{R}^p$ via the Johnson-Lindenstrauss lemma: $Y = \frac{1}{\sqrt{p}}XG^T$, where $G\sim \mathcal{N}(0, 1)^{p\times m}$ \\
    Select an $\frac{\epsilon}{6T}$-DP candidate set $C$ following Algorithm 3 of \citet{pmlr-v70-balcan17a}.\\
    Select an $\frac{\epsilon}{6T}$-DP subset  $\{u_1, \dots, u_k\}\subset C$ using Algorithm 4 of \citet{pmlr-v70-balcan17a}.\\
    Partition $Y$ into $S_j = \{i: j= \argmin_{l} || y_i - u_l|| \}, \, j=1, ...,k$.\\
    Recover $z_j^{(t)} = \texttt{sparse\_recovery} (\{x_i\}_{i \in S_j}, j=1, ...,k, \epsilon, \delta$) \\
    }
Choose $z_1, ...,z_k$ by sampling $Z$ from $Z^{(1)}, Z^{(2)}, ..., Z^{(T)}$ with probability proportional to $\exp \left( -\frac{\epsilon \mathcal{L}(Z^{(t)})}{24 \Lambda^2} \right)$ \\
\Return{ $z_1, ...,z_k$ }\
 \caption{\texttt{priavate\_prototypes}($X, k, \epsilon, \delta$) \citep{pmlr-v70-balcan17a}. The subroutines in lines 3 and 4 depend on the choice of $\delta$. The overall algorithm is $\epsilon$-DP. }
  \label{alg:kmeans}
\end{algorithm*}

\begin{algorithm*}[h!]
\SetAlgoLined
\KwIn{data $X \in \mathcal{B}(0,\Lambda )\subseteq \mathbb{R}^{ n \times p}$, parameters $ \epsilon, \delta$, initial cube $Q$ s.t. $\{x_i \}_{i=1}^n \subseteq Q$}
\KwResult{Private Grid $C \subseteq \mathds{R}^p$}
    Initialize depth a = 0, active set of cubes $\mathcal{A} = \{ Q \}$, and set $C = \emptyset $\\
     \While{$a \leq  n$  {\bf and} $\mathcal{A} \neq $ }{
     $a = a+1$ \\
     $C=C \cup \left( \cup_{Q_i \in \mathcal{A}}\text{center}(Q_i) \right)$ \\
         \For{$Q_i \in \mathcal{A}$}{
             Remove $Q_i$ from $\mathcal{A}$\\
              Partition $Q_i$ evenly in each dimension and obtain $2^p$ cubes $\{Q_i^{(l)} \}_{l=1}^{2^p}$\\
             \For{$l \in \{1,2, \dots, 2^p \}$ }{
                 Add $Q_i^{(l)}$ to $\mathcal{A}$ with probability $f\left( |Q_i^{(l)} \cap X| \right) $ where \\
                 $f(m) = \begin{cases} \frac{1}{2}\exp{-\epsilon'(\gamma - m))} & m \leq \gamma \\
                 1-\frac{1}{2} \exp{\epsilon'(\gamma - m))},& \text{otherwise} \end{cases}$ \\
                 $\epsilon' = \frac{\epsilon}{2 \log n}$ and $\gamma = \frac{20}{\epsilon'}\log \frac{n}{\delta}$ \\
             }
        }
    }
    \Return{C}
 \caption{\texttt{private\_partition} $ (X, \epsilon, \delta, Q$) \citep{pmlr-v70-balcan17a} }
  \label{alg:partition}
\end{algorithm*}

\begin{algorithm*}[h!]
\SetAlgoLined
\KwIn{data $X \in \mathcal{B}(0,\Lambda) \subseteq \mathbb{R}^{ n \times p}$, parameters $ \epsilon, \delta$}
\KwResult{Candidate center set $C \subseteq \mathbb{R}^{ \cdot \times p}$}
    Initialize  $C = \emptyset $\\
\For{$t = 1,2, \dots T = 25 k\log \frac{n}{\delta}$}{
    Sample shift vector $v \sim \mathcal{U}([-\Lambda, \Lambda]^p)$ \\
    Let $Q_v = [-\Lambda, \Lambda]^p + v$\\
    $C = C \cup \texttt{private\_partition} (X, \frac{\epsilon}{T}, \frac{\delta}{T}, Q_v)$
 }
\Return{C}
 \caption{\texttt{candidate} $ (X, \epsilon, \delta)$ \citep{pmlr-v70-balcan17a} }
  \label{alg:candidate}
\end{algorithm*}

\begin{algorithm*}[!h]
\SetAlgoLined
\KwIn{data $X \in \mathcal{B}(0,\Lambda)\subseteq \mathbb{R}^{ n \times p}$, parameters $ \epsilon, \delta$, Candidate set $C\subseteq \mathbb{R}^{ \cdot \times p}$}
\KwResult{Clustering centers $Z = [z_1,z_2, ...,_k] \subseteq C$}
    Uniformly sample $k$ centers \textit{i.i.d.} from $C$ and form $Z^{(0)}$\\
    $T \leftarrow \frac{n}{\delta}$\\
    \For{$t=1,2,..,T$}{
    Choose $x \in Z^{(t-1)}, y \in C \setminus Z^{(t-1)}$ with probability proportional to $\exp{-\epsilon\frac{\mathcal{L}(Z') - \mathcal{L}(Z^{(t-1)} }{8\Lambda^2(T+1)}}$ \\ 
    where $Z' = Z^{(t-1)} - \{x \} + \{ y \}$ \\
    $Z^{(t)} \leftarrow Z^{(t-1)}- \{x \} + \{ y \}$
    }
    Choose $t \in \{1,2,\dots, T$ with probability in proportion to $\exp{\frac{\epsilon \mathcal{L}(Z^{(t)}}{8(T+1)\Lambda^2}}$

\Return{$Z^{(t)}$}
 \caption{\texttt{localswap} $ (X, C, \epsilon, \delta)$ \citep{pmlr-v70-balcan17a} }
  \label{alg:localswap}
\end{algorithm*}

\newpage
\section{Experiments Details}

\subsection{Experiment 1: Private $k$-means vs $k$-means on Poisson Distributed Data} \label{sec:appendixExp1}

For this experiment we generated $U \sim \text{Norm}(0,1) \in \mathds{R}^{m,l}$, $V \sim \text{Norm}(0,1) \in \mathds{R}^{n,l}$, $\lambda = UV^T$ and $X \sim \text{Pois} (UV^T) $. We set $m = 100,000$, $n=500$, $l = 100$ and observe average behaviour of private $k$-means.  As $\epsilon$ increases, the level of privacy decreases thus reducing $k$-means objective and approaching the objective achieved by standard, non-private k-means.

                \begin{figure}[h!]
                \centering
\includegraphics[width=.5\textwidth]{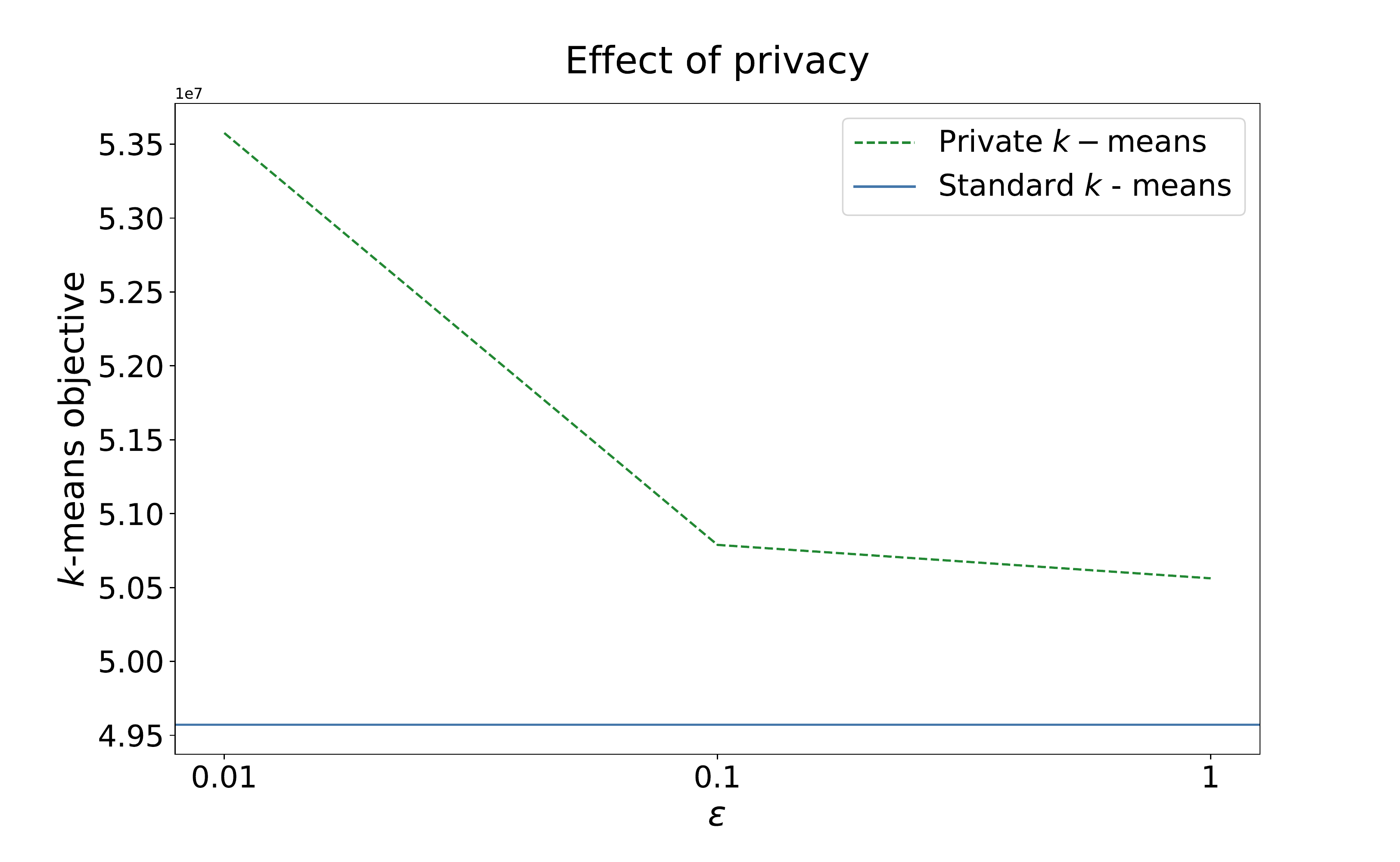}
\caption{$k$-means objective vs. level of privacy. As $\epsilon$ decreases, private $k$-means approaches the objective of non-private $k$-means.}
\label{fig:exp1:means}
\end{figure}

To implement standard $k$-means we used the Python library \texttt{scikit-learn}. For  private prototypes, we modified and implemented in Python publicly available MATLAB code  from  \cite{pmlr-v70-balcan17a} (\url{https://github.com/mouwenlong/dp-clustering-icml17}). 

\begin{figure*}[h!]
\centering
     \begin{subfigure}{0.495\textwidth}
         \centering
         \includegraphics[width=\textwidth]{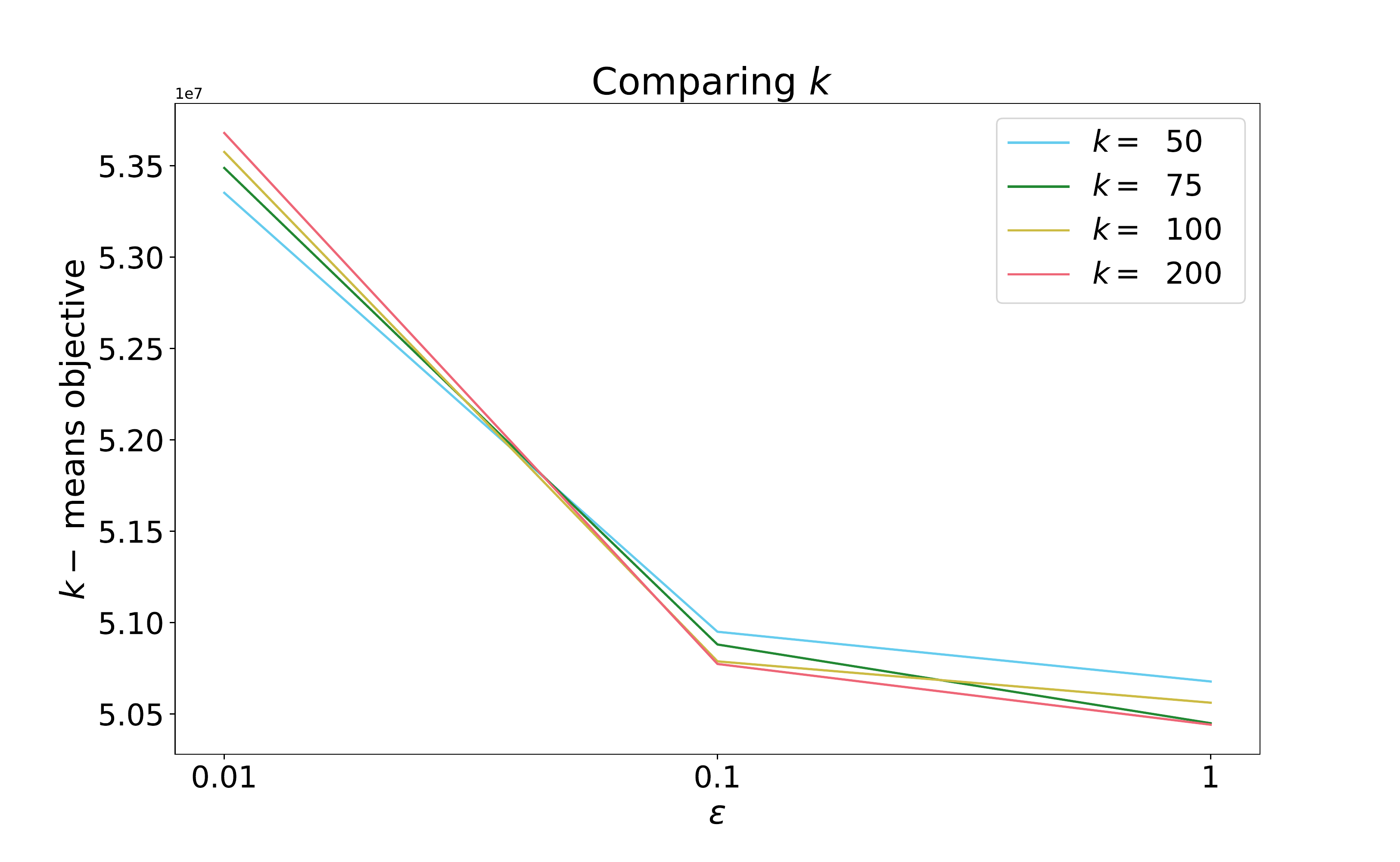}
        \caption{$k$-means loss vs. $\epsilon$ for different values of $k$. }
        \label{fig:app:exp11}
     \end{subfigure}
     \hfill
     \begin{subfigure}{0.495\textwidth}
         \centering
         \includegraphics[width=\textwidth]{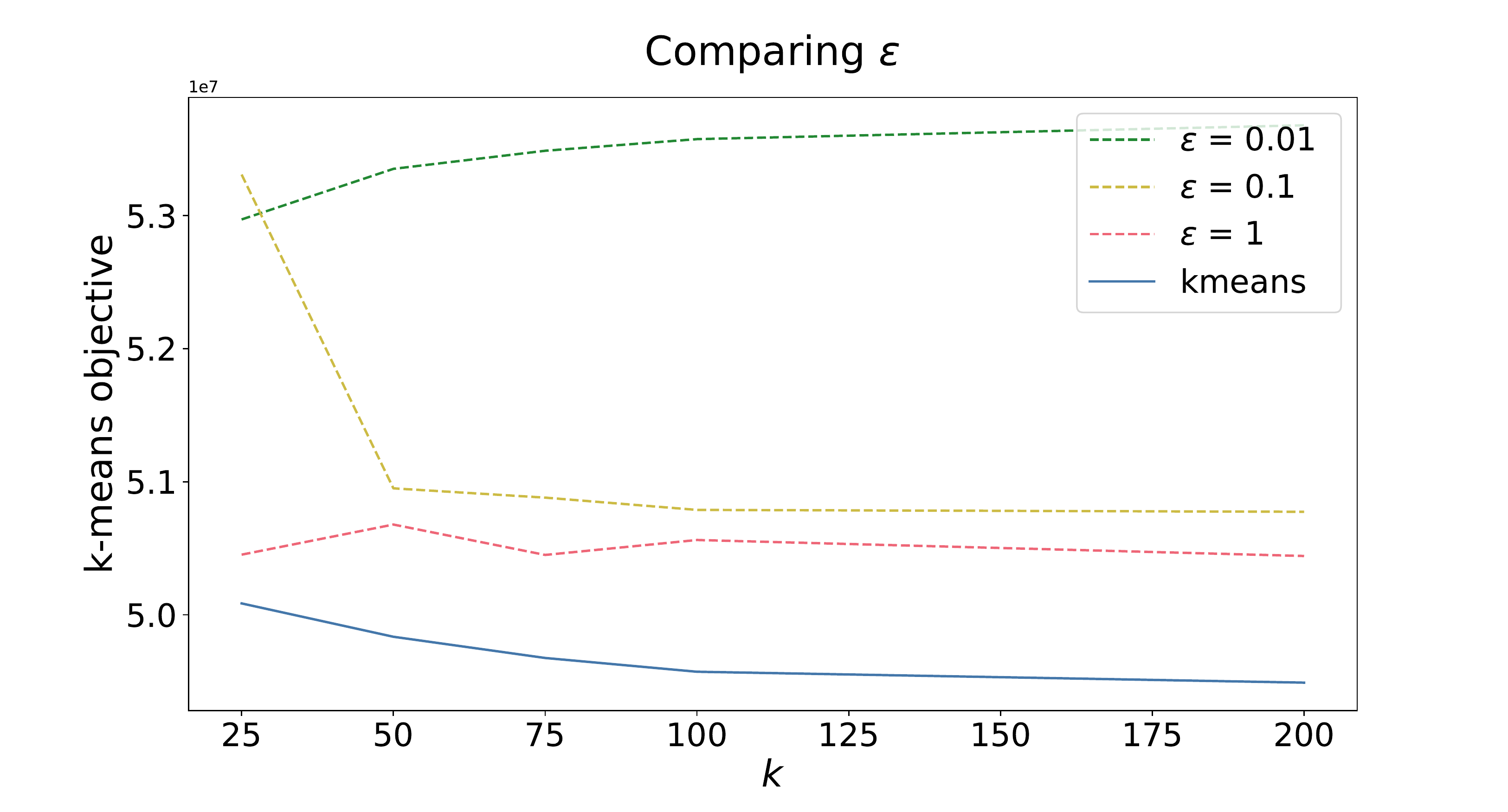}
\caption{$k$-means loss for different values of $\epsilon$. }
\label{fig:app:exp12}
     \end{subfigure}
     \caption{Private $k$-means on synthetic data. Larger values of $\epsilon$, i.e. less privacy, decrease the loss value. A large $k$ does not necessarily result in better performance. As shown in subfigure , for larger values of means, the private $k$-means algorithm repeats centers instead of overfitting, and objective minimization is stalled. }
     \label{fig:app:exp1}
\end{figure*}

\subsection{Further Experimentation on the Number of Entities}
\fref{fig:numentities}  shows the RMSE on the synthetic test dataset described in section 5.1. We observe that as the number of entities increases, the convergence improves.  This is expected since the number of observations used to approximate V also grows.

\begin{figure}[h!]
\centering
\includegraphics[width = 0.5\textwidth]{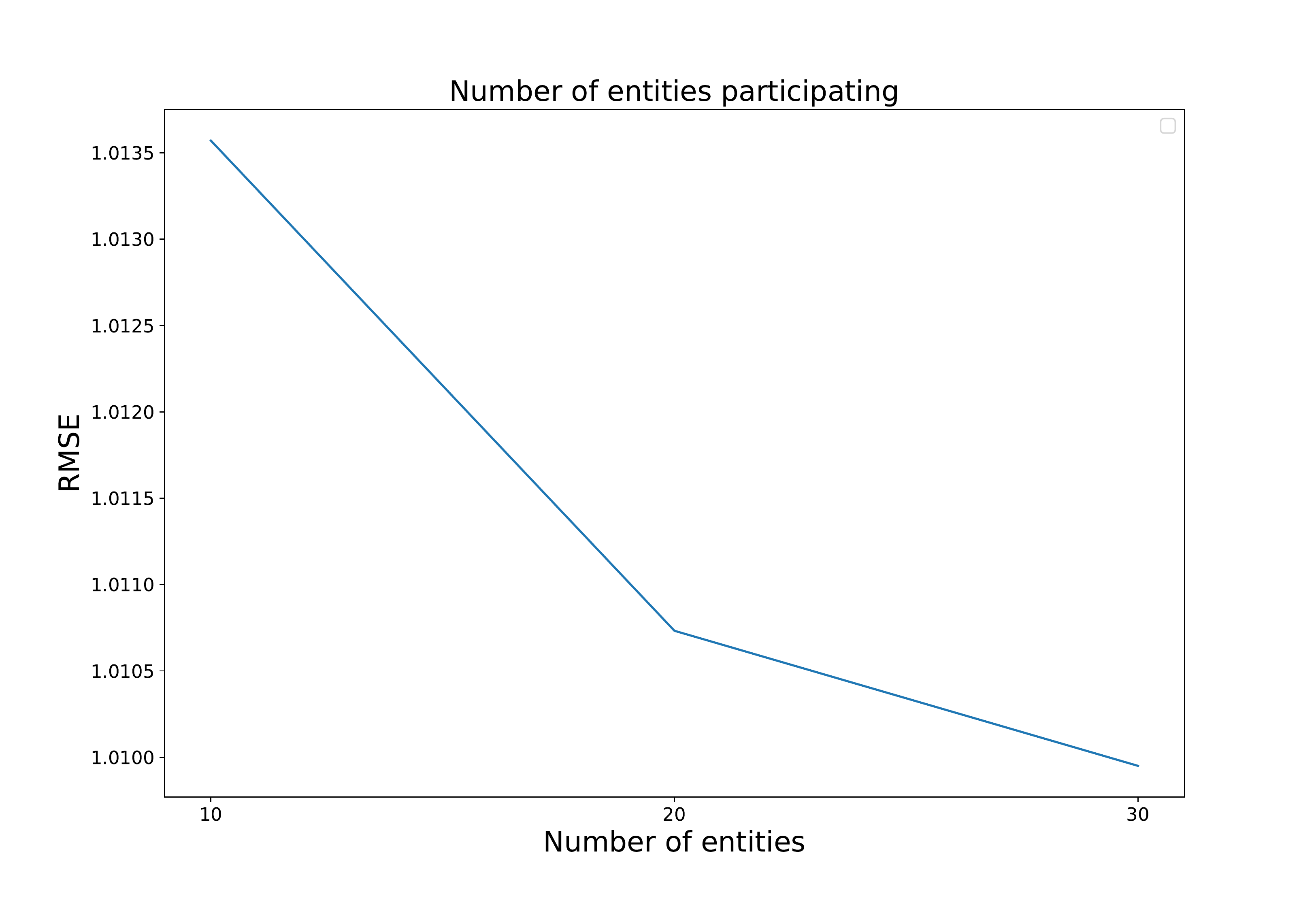}
\caption{Convergence of matrix factorization for different number of entities }
\label{fig:numentities}
\end{figure}

\subsection{ Experiment 2: Varying Parameters for Normal Synthetic Data}
\label{sec:appendizExp2}
        \begin{figure}[h!]
\includegraphics[width=\textwidth,left]{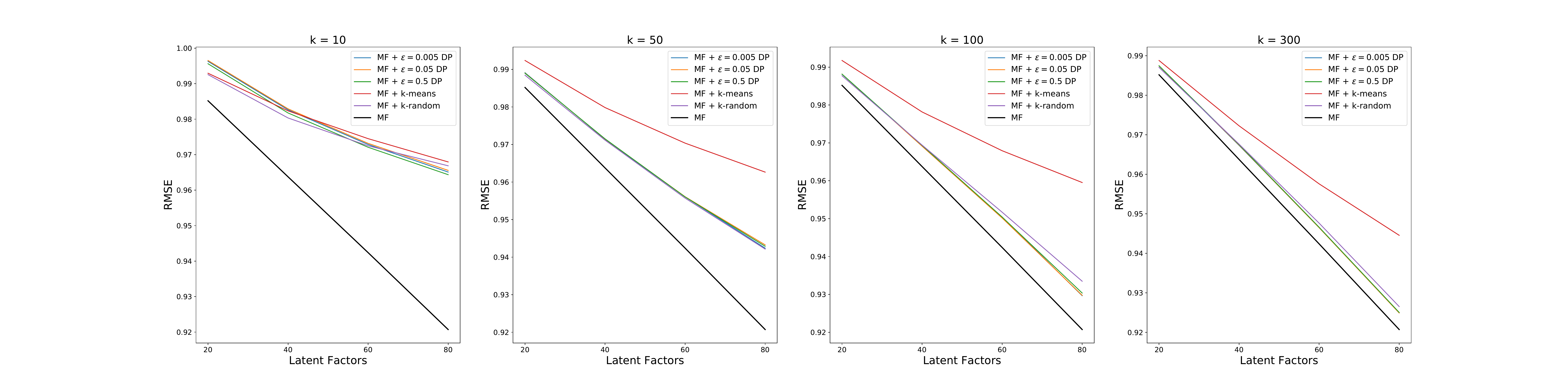}
\caption{Comparison of different prototype methods. As $k$ and $\ell$ increase, $k$-random exemplars and private $k$-means maintain competitive performance.}
\label{fig:exp2:app1}
\end{figure}

        \begin{figure}[h!]
\includegraphics[width=\textwidth,left]{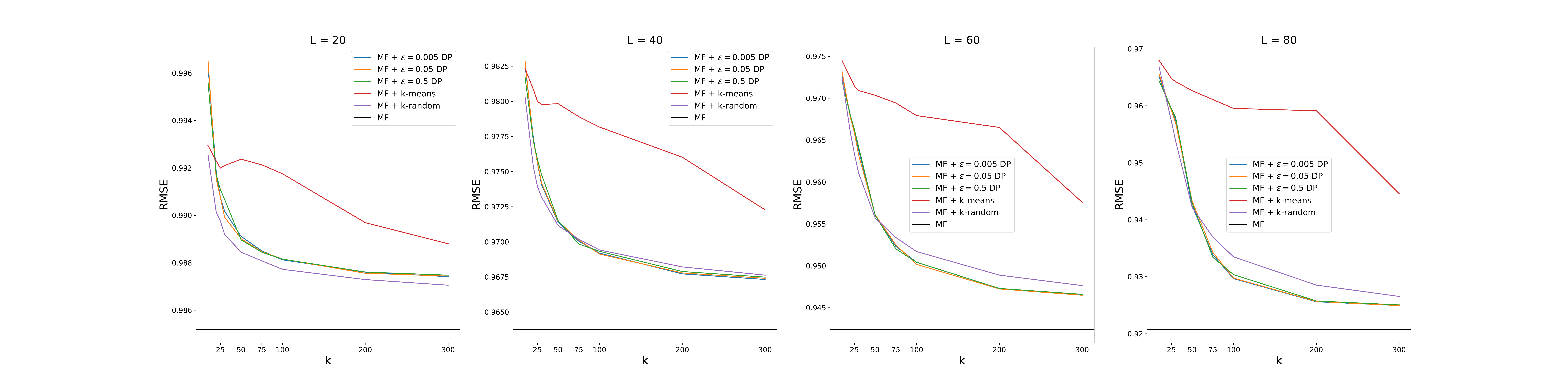}
\caption{Comparison of various methods for different values of $L$. Private methods have superior performance for large $\ell$.}
\label{fig:exp2:app3}
\end{figure}

In Section 5.3 we showed results for fixed values of the number of prototypes $k$ and the number of latent features $\ell$. Below we show additional plots for different values of those parameters. 

 In \fref{fig:exp2:app1} we observe as the number of samples increases, random $k$-exemplars outperforms $k$-means for all values of $\ell$. Note that private $k$-means performs well over a wide range of $k$. As $k$ increases, private $k$-means converge to the same value for various values of $\epsilon$.  \fref{fig:exp2:app3} compares all methods for different values of $k$. The difference in RMSE is clearer for small values of  $k$. For large values of $k$, the performance of $k$-random and $k$-private approaches that of matrix factorization.

\newpage
\subsection{Movielens Results}
\begin{figure*}[h!]
\centering
     \begin{subfigure}{0.45\textwidth}
         \centering
         \includegraphics[width=\textwidth]{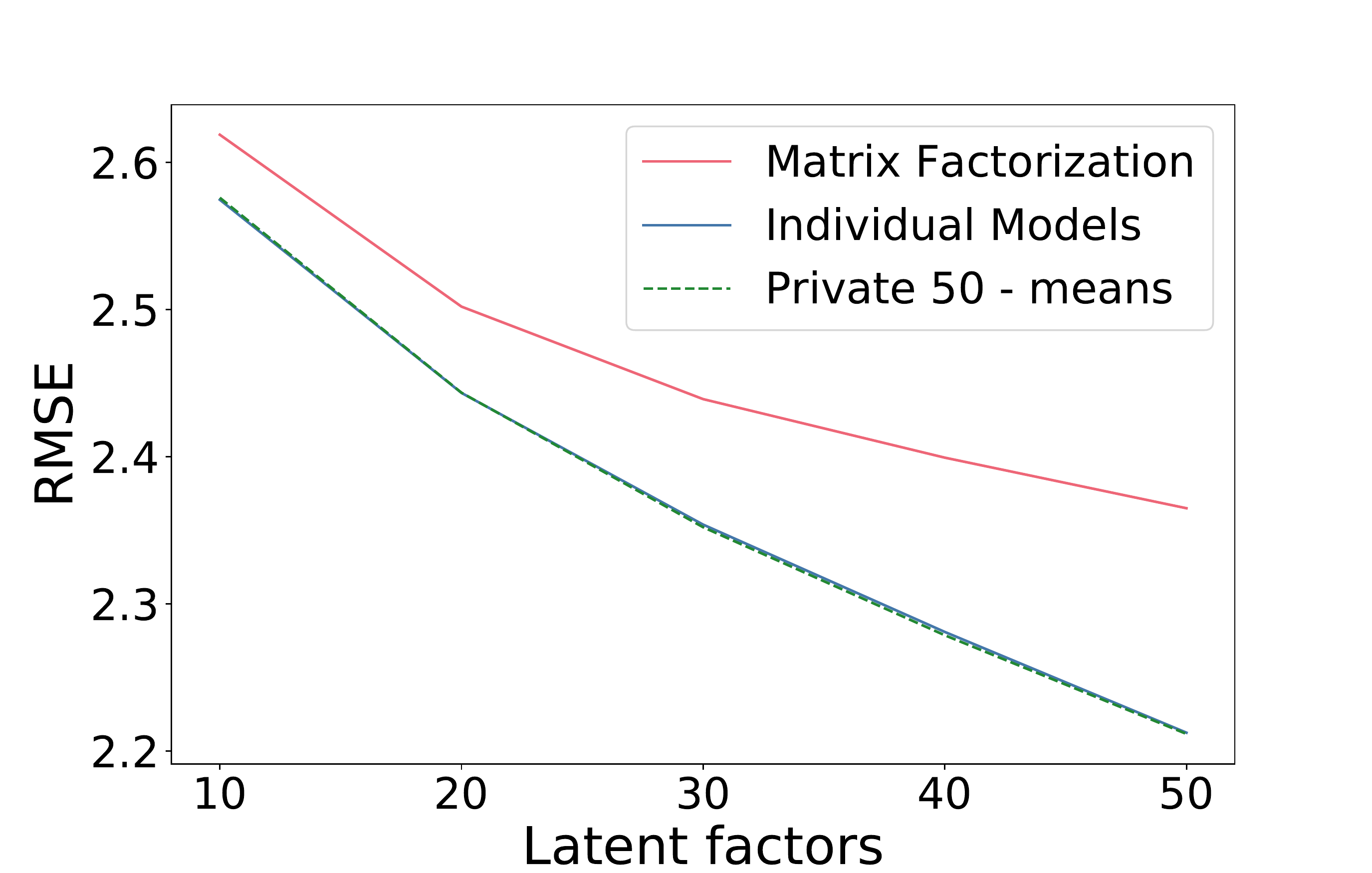}
        \caption{averaged RMSE on train data for the Movielens 1M dataset. }
        \label{fig:app:exp4RMSEtrain}
     \end{subfigure}
     \hfill
     \begin{subfigure}{0.45\textwidth}
         \centering
         \includegraphics[width=\textwidth]{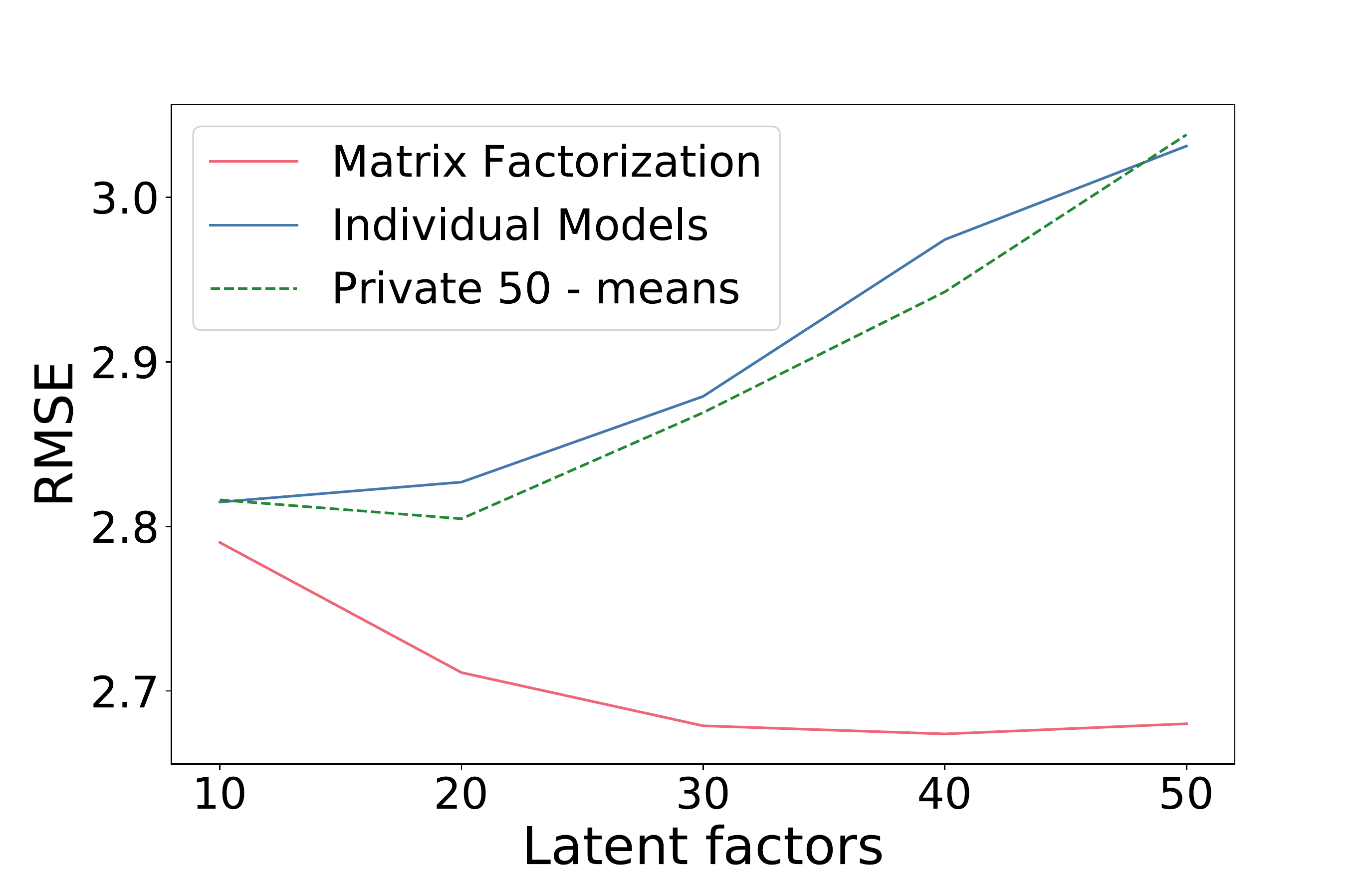}
\caption{averaged RMSE on test data for the Movielens 1M dataset. }
\label{fig:fed_learn_test_movies}
     \end{subfigure}
     \caption{ }
     \label{fig:app:exp4RMSE}
\end{figure*}

        \begin{figure}[h!]
        \centering
\includegraphics[width=0.5\textwidth]{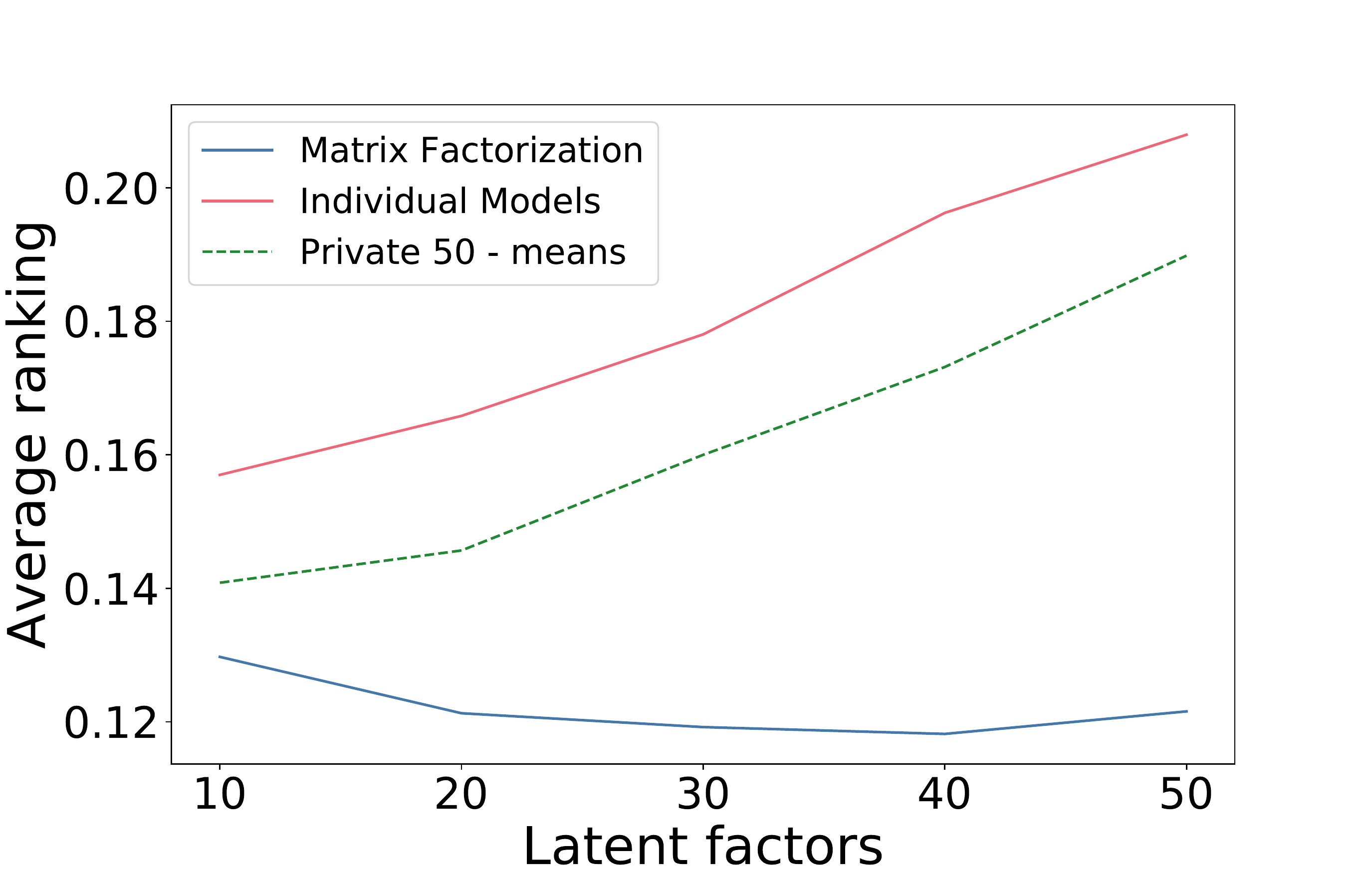}
\caption{Average rank on the Movielens 1M dataset. Privacy deteriorates performance, however DP-prototypes allow entities to collaborate and improve recommendations. }
\label{fig:app:exp3avgRang}
\end{figure}

Similar to experiments on the eICU dataset, we observe unsurprisingly in \fref{fig:app:exp4RMSE} that the global non-private matrix factorization model has lower RSME than the distributed approaches (i.e., individual models and Private 50-means). However, there is a benefit from collaboration. Recall that an average ranking above 0.5 means a ranking no better than random. Conversely, lower values indicate highly ranked recommendations matching the users’ patterns.  We observe in \fref{fig:app:exp4RMSE} the benefit of collaboration: the quality of recommendations is better for the prototypes models than the local individual models. With a small privacy budget, our method is able to share insights among entities, without sacrificing their privacy, and delivering better recommendations.

\end{document}